\documentclass[11pt]{article}

\usepackage[margin=1in]{geometry}
\usepackage{amsmath,amssymb,amsthm,mathtools}
\usepackage{bbm, comment}
\usepackage{enumitem}
\usepackage{mathrsfs}
\usepackage{hyperref}
\usepackage{soul}
\usepackage[numbers,sort&compress]{natbib}

\hypersetup{
  colorlinks=true,
  linkcolor=blue,
  citecolor=blue,
  urlcolor=blue
}

\newtheorem{theorem}{Theorem}[section]

\newtheorem{lemma}[theorem]{Lemma}
\newtheorem{corollary}[theorem]{Corollary}

\newcommand{\Pbb}{\mathbb{P}}
\newcommand{\Ebb}{\mathbb{E}}

\newcommand{\cov}{{\textnormal{cov}}}

\newcommand{\diag}{{\textnormal{diag}}}

\newcommand{\blue}{\color{blue}}
\newcommand{\red}{\color{red}}
\newcommand{\black}{\color{black}}

\newif\ifva
\vatrue

\begin{document}

\title{True Self-Avoiding Walk for Accelerating 
Markov-Chain \\ Monte Carlo
Integration}

\author{Qinghua (Devon) Ding, \quad Venkat Anantharam\\
    Department of Electrical Engineering and Computer Sciences\\
    University of California at Berkeley\\
    Berkeley, CA, United States\\
    Email: \{devon\_ding, ananth\}@berkeley.edu}

\maketitle

\begin{abstract}
We study true self-avoiding walk (TSAW) as a mechanism for improving empirical integral estimation via Markov chain Monte Carlo
(MCMC).
We consider finite-state adaptive sampling dynamics associated with 
an irreducible Markov kernel $P$ on a finite set,
with stationary distribution $\pi$, in which 
the
transition probabilities are 
penalized
according to empirical overuse.

Our main result is that the empirical occupation counts $L_t(i)$ and transition counts $N_t(i,j)$ of the resulting TSAW-based walk satisfy
\[
L_t(i)-t\pi_i = O(\sqrt{\log t})
\quad\text{and}\quad
N_t(i,j)-t\pi_iP_{ij}=O(\sqrt{\log t})
\qquad\text{almost surely}
\]
for every state $i$ and every edge $(i,j)$ with $P_{ij}>0$. Consequently, for every bounded function $f:V\to\mathbb R$, the error of our integral estimator converges as
\[
\left|\frac1t\sum_{s=0}^{t-1} f(X_s)-\sum_{i\in V}\pi_i f(i)\right|
=
O\left(\frac{\sqrt{\log t}}{t}\right)
\qquad\text{almost surely}.
\]

These results show that, in contrast with the usual $t^{-1/2}$ error scaling for empirical averages under standard random-walk-based methods, TSAW-based estimator yields empirical integral errors of order $O(\sqrt{\log t}/t)$ almost surely, thereby achieving a substantially sharper dependence on the sample size $t$.
\end{abstract}

\section{Introduction}
\label{sec:intro}

Let $V$ be a finite state space, let $P$ be an irreducible Markov kernel on $V$, and let $\pi$ denote its unique stationary distribution. For a bounded function $f:V\to\mathbb{R}$, a basic task in the Markov chain Monte Carlo is to estimate the stationary integral
\[
\pi(f):=\sum_{x\in V}\pi(x)f(x)
\]
from a trajectory $(X_t)_{t\ge 0}$ of the chain, then the empirical average converges as $T\rightarrow \infty$,
\[
\frac{1}{t}\sum_{s=0}^{t-1} f(X_s) \rightarrow \pi(f).
\]
The finite-time quality of this estimator is governed by the empirical occupation statistics of the trajectory. Writing the vertex local time as
\[
L_t(i):=\sum_{s=0}^{t-1}\mathbf{1}\{X_s=i\},
\qquad i\in V,
\]
one has
\[
\frac{1}{t}\sum_{s=0}^{t-1} f(X_s)-\pi(f)
=
\frac{1}{t}\sum_{i\in V} f(i)\bigl(L_t(i)-t\pi_i\bigr).
\]
Thus empirical integral estimation is directly controlled by the discrepancy between the occupation counts $L_t(i)$ and their target values $t\pi_i$.

For ordinary random-walk-based MCMC, the trajectory may spend a significant amount of time revisiting recently explored regions of the state space. This produces local oversampling, slows the correction of occupation imbalance, and ultimately limits the quality of finite-time empirical averages. A substantial literature addresses this phenomenon by modifying the base chain to avoid backtracking or to break reversibility; see, for example, \cite{DiaconisHolmesNeal2000,Neal2004,AlonBenjaminiLubetzkySodin2007,LeeXuEun2012,TuritsynChertkovVucelja2011,ChenHwang2013,MaChenWuFox2016,ThinEtAl2020}. In a different direction, nonlinear and self-interacting Markov chains provide a general framework in which the transition rule depends on the occupation history of the process \cite{DelMoralMiclo2004,DelMoralMiclo2006,AndrieuJasraDoucetDelMoral2007,AndrieuJasraDoucetDelMoral2011,DelMoralDoucet2010,FortMoulinesPriouret2011,Fort2015}.

The present paper studies true self-avoiding walk (TSAW) as a mechanism for reducing this occupation imbalance. Classical TSAW and weakly self-avoiding walk models were introduced and studied mainly for path properties such as scaling behavior, recurrence, and self-repulsion on lattice-like state spaces \cite{AmitParisiPeliti1983,Toth1995,VetoToth2008,Grassberger2017}. More recently, Doshi, Hu, and Eun \cite{DoshiHuEun2023} developed self-repellent random walks on general finite graphs as nonlinear Markov chains targeting a prescribed stationary distribution, proved almost sure convergence of the empirical distribution, and established a central limit theorem whose asymptotic covariance decreases as the repulsion strength increases. Their work provides a natural benchmark for self-avoiding MCMC on finite state spaces.

We design an estimator based on TSAW that improves the convergence rate in terms of the length of the trajectory, compared to the self-repellent walk based estimator. 
Specifically, we show a better bound on the finite-time discrepancy for a TSAW-based dynamics. 
Our main theorem shows that the empirical occupation counts and empirical transition counts remain within $O(\sqrt{\log t})$ of their target values almost surely. More precisely, if
\[
L_t(i):=\sum_{s=0}^{t-1}\mathbf{1}\{X_s=i\},
\qquad
N_t(i,j):=\sum_{s=0}^{t-1}\mathbf{1}\{X_s=i,\ X_{s+1}=j\},
\]
then for every state $i\in V$ and every edge $(i,j)$ with $P_{ij}>0$, we show that our estimators have discrepancy bounded as
\[
L_t(i)-t\pi_i = O(\sqrt{\log t})
\qquad\text{and}\qquad
N_t(i,j)-t\pi_iP_{ij}=O(\sqrt{\log t}),
\]
with high probability. As an immediate consequence, for every observable $f:V\to\mathbb{R}$,
\[
\left|
\frac{1}{t}\sum_{s=0}^{t-1} f(X_s)-\pi(f)
\right|
=
O\left(\frac{\sqrt{\log t}}{t}\right),
\]
with high probability. This gives an explicit finite-time error bound for empirical integral estimation. 
By Corollary~4.6 of \cite{DoshiHuEun2023}, for every observable $f:V\to\mathbb{R}$ and every repulsion parameter $\alpha\ge 0$, the SRRW empirical estimator converges almost surely to $\pi(f)$ and satisfies the central limit theorem
\[
\sqrt{t}\left(\frac1t\sum_{s=0}^{t-1} f(X_s)-\pi(f)\right)
\ \xrightarrow[]{\mathrm{dist}}\ 
\mathcal N\bigl(0,\sigma^2_{f,\alpha}\bigr),
\]
for a variance $\sigma^2_{f,\alpha}$ depending on $f$ and $\alpha$. Thus, in \cite{DoshiHuEun2023}, the empirical integral error remains at the standard $O\big(\frac{1}{\sqrt{t}}\big)$ scale, although the asymptotic variance decreases as the repulsion parameter $\alpha$ increases. By contrast, our main theorem yields the $O\big(\frac{\sqrt{\log t}}{t}\big)$ bound.
In this sense, we improve the convergence rate in terms of $t$ by a factor of $O\big(\sqrt{\frac{t}{\log t}}\big)$.

The argument proceeds in the following steps. We begin with a warm-up model on the star graph, where the TSAW dynamics can be represented through an exponential-race embedding, allowing us to analyze the cover time in detail. We then study a finite-alphabet self-balancing model and construct a Lyapunov function for the occupation-deficit process, from which we obtain uniform exponential moment bounds by a standard drift argument. Finally, we pass to the finite-state Markov-kernel setting by applying the same balancing mechanism row-wise to the transition counts, and combine the resulting row-wise estimates with a Poisson equation for the occupation vector to derive the almost-sure $O(\sqrt{\log t})$ discrepancy bound stated above.

The remainder of the paper is organized as follows. Section~\ref{sec:related-work} discusses related work. Section~\ref{sec:faster-cover} studies the warm-up star-graph model. Section~\ref{sec:self_balancing} develops the self-balancing framework, first in the finite-alphabet setting and then for general finite irreducible Markov kernels.

\section{Related Work}
\label{sec:related-work}

The closest prior work to ours is that of Doshi, Hu, and Eun \cite{DoshiHuEun2023}, who study self-repellent random walks on finite graphs as nonlinear Markov chains targeting a prescribed stationary distribution. In their setting, the empirical distribution converges almost surely to the target law, and the empirical integral estimator satisfies a central limit theorem; see in particular Corollary~4.6 of \cite{DoshiHuEun2023}. Moreover, Corollary~4.7 of \cite{DoshiHuEun2023} shows that increasing the repulsion strength decreases the asymptotic variance. Thus the improvement in \cite{DoshiHuEun2023} is at the level of the asymptotic variance constant, while the empirical integral error remains at the standard $O(1/\sqrt{t})$ scale. 

True self-avoiding walk and weakly self-avoiding walk have a much longer history in probability theory and statistical physics. Classical works such as \cite{AmitParisiPeliti1983,Toth1995,VetoToth2008,Grassberger2017} study self-repelling walks mainly through path properties, including scaling limits, recurrence, and self-repulsion on lattice-type graphs. These works provide the probabilistic background for TSAW, but they are not formulated in terms of empirical integral estimation for a prescribed target distribution on a finite state space. Our use of TSAW is different in purpose: the self-avoidance mechanism is introduced here in order to control occupation imbalance relative to a target law.

There is also a substantial literature showing that non-backtracking or nonreversible modifications can improve the efficiency of random-walk-based Monte Carlo methods; see, for example, \cite{DiaconisHolmesNeal2000,Neal2004,AlonBenjaminiLubetzkySodin2007,LeeXuEun2012,TuritsynChertkovVucelja2011,ChenHwang2013,MaChenWuFox2016,ThinEtAl2020}. These methods typically improve performance by preventing immediate reversals or by introducing irreversible flow into the dynamics. Our setting is different: the transition rule depends on empirical usage accumulated along the trajectory, and the main quantity of interest is the resulting occupation and transition discrepancy. In this sense, the present paper is closer to trajectory-dependent self-avoidance than to nonreversible acceleration.

From a broader perspective, the model studied here belongs naturally to the theory of nonlinear and self-interacting Markov chains. In this direction, Del Moral and Miclo \cite{DelMoralMiclo2004,DelMoralMiclo2006}, Andrieu, Jasra, Doucet, and Del Moral \cite{AndrieuJasraDoucetDelMoral2007,AndrieuJasraDoucetDelMoral2011}, Del Moral and Doucet \cite{DelMoralDoucet2010}, and Fort and coauthors \cite{FortMoulinesPriouret2011,Fort2015} develop general frameworks for Markov dynamics whose transition mechanism depends on occupation statistics or interacting empirical measures. Our contribution fits into this history-dependent framework, but focuses on a concrete TSAW-based construction for which one can prove explicit almost-sure bounds on vertex and edge discrepancies.

Taken together, these works place the present paper at the intersection of self-avoiding random walks, nonlinear Markov chains, and MCMC variance reduction. Relative to the existing literature, the main novelty here is the explicit finite-time pathwise control of empirical occupation and transition errors, and the resulting almost-sure bound of order $O(\sqrt{\log t}/t)$ for empirical integral estimation.

\section{A warm-up: cover time of TSAW on the star graph}
\label{sec:faster-cover}

We begin with a warm-up model that isolates the effect of self-avoidance on exploration. Let $G_n$
be the star graph with hub $0$ and leaves $1,\dots,n$. Under TSAW, every jump from the hub to a leaf
is followed by an immediate return to the hub, so the dynamics is completely determined by the sequence
of hub departures. It is therefore natural to study the reduced process obtained by recording only the
successive choices of leaves at the hub.

Fix $\rho\in(0,1)$. For $m\ge 0$ and $i\in[n]:=\{1,\dots,n\}$, let
\[
C_m(i):=\sum_{r=1}^m \mathbf 1\{I_r=i\}
\]
be the number of times leaf $i$ has been chosen in the first $m$ hub departures. 
Here $C_0(i) = 0$ for all $i$.
We define the reduced
TSAW dynamics by
\begin{equation}\label{eq:star_rule}
\mathbb P(I_{m+1}=i\mid I_1,\dots,I_m)
=
\frac{\rho^{\,C_m(i)}}{\sum_{j=1}^n \rho^{\,C_m(j)}},
\qquad i\in[n],
\end{equation}
with $I_1$ being chosen uniformly at random.
This is exactly the hub-departure chain induced by TSAW on the star graph, after absorbing the
deterministic factor coming from the return step into the parameter $\rho$.
Let
\[
\tau_{\cov}^{(n)}:=\inf\bigl\{m\ge 0:\ C_m(i)\ge 1\ \text{for all }i\in[n]\bigr\}
\]
be the cover time, measured in hub departures. The usual discrete-time cover time differs only by a
factor of $2$. 

For random variables $Y_n$ and a deterministic constant $y$, we write
\(
Y_n \xrightarrow{\mathbb P} y
\)
if for every $\varepsilon>0$,
\[
\mathbb P(|Y_n-y|>\varepsilon)\to 0
\qquad\text{as }n\to\infty.
\]
We also write $R_n=o_{\mathbb P}(1)$ if $R_n\xrightarrow{\mathbb P}0$.

Our main result in this section is the following.

\begin{theorem}[cover time of TSAW on the star graph]
\label{thm:star_cover}
Fix $\rho\in(0,1)$ and set
\(
c_\rho:=\frac{1}{\log(1/\rho)}.
\)
Let $\tau_{\cov}^{(n)}$ be the cover time, measured in hub departures, of the TSAW dynamics on the
star graph with $n$ leaves and effective discount parameter $\rho$. Then for every $\varepsilon>0$ and
every $\delta>0$, there exists $n_0=n_0(\varepsilon,\delta,\rho)$ such that for all $n\ge n_0$,
\[
\mathbb P\left(
\left|
\frac{\tau_{\cov}^{(n)}}{n\log\log n}-c_\rho
\right|
\le \varepsilon
\right)
\ge 1-\delta.
\]
\end{theorem}

The proof is based on an exponential-race embedding. For each leaf $i\in[n]$, let
$\{E_{i,k}:k\ge 0\}$ be independent random variables with
\[
E_{i,k}\sim \mathrm{Exp}(\rho^k),
\qquad i\in[n],\ k\ge 0,
\]
and assume that these families are mutually independent across $i$. Define the arrival times
\[
S_{i,r}:=\sum_{k=0}^{r-1} E_{i,k},
\qquad r\ge 1,
\qquad
S_{i,0}:=0,
\]
and the associated birth processes
\[
N_i(s):=\max\{r\ge 0:\ S_{i,r}\le s\},
\qquad s\ge 0.
\]
Let
\[
K(s):=\sum_{i=1}^n N_i(s)
\]
be the total number of births by time $s$, and let
\[
M:=\max_{1\le i\le n} S_{i,1}=\max_{1\le i\le n} E_{i,0}
\]
be the time at which the last leaf receives its first birth.

The first lemma shows that this continuous-time system reproduces the discrete TSAW dynamics exactly.

\begin{lemma}[exponential-race embedding]
\label{lem:star_embedding}
The reduced TSAW process \eqref{eq:star_rule} can be coupled with the birth processes
$\{N_i(s)\}_{i=1}^n$ so that the ordered sequence of births in the merged timeline
$\{(i,r): S_{i,r}\}_{i\in[n],\,r\ge1}$ has the same law as the sequence of hub departures
$(I_m)_{m\ge1}$. Under this coupling,
\(
\tau_{\cov}^{(n)}=K(M)
\)
almost surely, where $M=\max_{1\le i\le n} S_{i,1}$. 
For every $\varepsilon\in(0,1)$, if $C_\varepsilon:=\log(2/\varepsilon)$, then for all
sufficiently large $n$,
\[
\mathbb P\bigl(|M-\log n|\le C_\varepsilon\bigr)\ge 1-\varepsilon.
\]
On the event $\{|M-\log n|\le C_\varepsilon\}$, we have
\[
K(\log n-C_\varepsilon)\le \tau_{\cov}^{(n)}\le K(\log n+C_\varepsilon).
\]
\end{lemma}

\begin{proof}
Suppose that after $m$ births the current birth counts are $(c_1,\dots,c_n)$. By the memoryless
property of the exponential distribution, the residual waiting time to the next birth in process $i$
is $\mathrm{Exp}(\rho^{c_i})$, independently across $i$. Therefore the next birth occurs in process
$i$ with probability $\rho^{c_i}/\sum_{j=1}^n \rho^{c_j}$, which is exactly the transition rule
\eqref{eq:star_rule}. Iterating proves that the merged birth order reproduces the reduced TSAW
dynamics.

For each $i\in[n]$, the birth process $N_i$ is non-explosive. Indeed, since $\rho^k\le 1$ for all
$k\ge0$, each waiting time $E_{i,k}\sim \mathrm{Exp}(\rho^k)$ stochastically dominates an
$\mathrm{Exp}(1)$ random variable. Therefore the partial sums
\[
S_{i,r}=\sum_{k=0}^{r-1}E_{i,k}
\]
tend to $+\infty$ almost surely as $r\to\infty$, and hence $N_i(s)<\infty$ almost surely for every
finite $s\ge0$. Since $M=\max_{1\le i\le n}E_{i,0}$ is the maximum of finitely many exponential
random variables, we also have $M<\infty$ almost surely. It follows that
\[
K(M)=\sum_{i=1}^n N_i(M)<\infty
\qquad\text{almost surely}.
\]

Since $M=\max_i S_{i,1}$ is the time by which every process has had its first birth, the number of
births by time $M$ is exactly the number of hub departures required to visit every leaf at least once.
Hence $\tau_{\cov}^{(n)}=K(M)$ almost surely. Because $K(\cdot)$ is nondecreasing, the sandwich
bound follows immediately.

To control $M$, note that $E_{1,0},\dots,E_{n,0}$ are i.i.d.\ $\mathrm{Exp}(1)$, so
$\mathbb P(M\le y)=(1-e^{-y})^n$ for $y\ge0$. If $x\in[0,\log n]$, then
\[
\mathbb P(M\le \log n-x)=\left(1-\frac{e^x}{n}\right)^n\le e^{-e^x},
\]
while
\[
\mathbb P(M\ge \log n+x)
=1-\left(1-\frac{e^{-x}}{n}\right)^n
\le e^{-x}.
\]
Here we used the basic inequalities $1-u\le e^{-u}$ and
$1-(1-u)^n\le nu$ for $u\in[0,1]$. Thus $\mathbb P(|M-\log n|>x)\le e^{-x}+e^{-e^x}$. Taking $x=C_\varepsilon=\log(2/\varepsilon)$ gives
$e^{-x}+e^{-e^x}\le \varepsilon/2+e^{-2/\varepsilon}\le \varepsilon$ (note that $ue^{-u}\leq 1$ for $u>0$), and the claim follows.
\end{proof}

Therefore, to bound the cover time it suffices to study $K(s_n)$ for some deterministic $s_n=\log n + O(1)$. And furthermore, to study $K(s_n)=\sum_{i=1}^n N_i(s_n)$, which is a sum of i.i.d. random variables $N_i(s_n)$'s, it suffices to study the individual distributions. This is reflected in the following lemmas.

\begin{lemma}[birth process with exponentially decaying rates]
\label{lem:one_leaf_birth}
Fix $\rho\in(0,1)$ and, for $s\ge1$, let $N(s)$ denote a generic copy of $N_i(s)$. Define
\[
i_*(s):=\left\lfloor \log_{1/\rho} s\right\rfloor+1,
\qquad\text{so that}\qquad
\rho^{\,i_*(s)-1}s\in[1,1/\rho].
\]
Set
\[
C_1(\rho):=e^{1/\rho},
\qquad
C_2(\rho):=\prod_{m=0}^{\infty}\frac{1}{1-\frac12\rho^m},
\]
\[
A_\rho:=1+\frac{C_1(\rho)}{1-\sqrt{\rho}}+\frac{2C_2(\rho)}{1-\rho},
\]
and
\[
B_\rho
:=
2\left(
C_1(\rho)\frac{\sqrt{\rho}(1+\sqrt{\rho})}{(1-\sqrt{\rho})^2}
+
2C_2(\rho)\frac{1+\rho}{(1-\rho)^2}
\right)
+
2\left(
\frac{C_1(\rho)}{1-\sqrt{\rho}}+\frac{2C_2(\rho)}{1-\rho}
\right)^2.
\]
Then, for every $s\ge1$, the following hold:

\begin{enumerate}
\item[\textnormal{(i)}] For every integer $k\ge0$,
\[
\mathbb P\bigl(N(s)\ge i_*(s)+k\bigr)
\le
C_1(\rho)\,\rho^{k(k+1)/2}
\le
C_1(\rho)\,\rho^{k/2}.
\]

\item[\textnormal{(ii)}] For every integer $k\ge1$,
\[
\mathbb P\bigl(N(s)\le i_*(s)-k\bigr)
\le
C_2(\rho)\exp\bigl(-\tfrac12\rho^{-k+1}\bigr),
\]
where the event is understood to be empty when $i_*(s)-k<0$.

\item[\textnormal{(iii)}] One has
\[
\left|
\mathbb E[N(s)]-\frac{\log s}{\log(1/\rho)}
\right|
\le A_\rho,
\qquad
\operatorname{Var}(N(s))\le B_\rho.
\]
\end{enumerate}
\end{lemma}

This lemma can be proved using moment generating functions. We defer this proof to Appendix \ref{apx:one_leaf_birth}.

\begin{lemma}[concentration of the total number of departures]
\label{lem:K_concentration}
Fix $\rho\in(0,1)$, set $c_\rho:=1/\log(1/\rho)$, and let $c\in\mathbb{R}$ be fixed. For each
$n$ sufficiently large, define $s_n:=\log n+c$. Then for every $\varepsilon>0$ and every $\delta>0$,
there exists $n_0=n_0(\varepsilon,\delta,\rho,c)$ such that for all $n\ge n_0$,
\[
\mathbb P\left(
\left|
\frac{K(s_n)}{n\log\log n}-c_\rho
\right|
\le \varepsilon
\right)\ge 1-\delta.
\]
\end{lemma}

The proof of this lemma uses the fact that for any deterministic sequence $(s_n)_{n\geq 1}$, $K(s_n)$ is the sum of $n$ i.i.d. random variables. The proof
\red
of
\black
this lemma can be found in Appendix \ref{sec:proof-k-concentration}.

\begin{proof}[Proof of Theorem~\ref{thm:star_cover}]
Fix $\varepsilon>0$ and $\delta>0$. We must prove that, for all sufficiently large $n$,
\[
\mathbb P\left(
\left|
\frac{\tau_{\cov}^{(n)}}{n\log\log n}-c_\rho
\right|
\le \varepsilon
\right)\ge 1-\delta.
\]

Set $\delta_1=\delta_2=\delta_3=\delta/3$ and let $C:=\log(2/\delta_1)$. By
Lemma~\ref{lem:star_embedding}, there exists $n_1$ such that for all $n\ge n_1$, the event
$E_n:=\{|M-\log n|\le C\}$ satisfies $\mathbb P(E_n)\ge 1-\delta_1$. On $E_n$, we have the
deterministic sandwich
\[
K(\log n-C)\le \tau_{\cov}^{(n)}\le K(\log n+C).
\]

Now define $s_n^-:=\log n-C$ and $s_n^+:=\log n+C$. Since $C$ is fixed, both sequences satisfy
$s_n^\pm=\log n+O(1)$ and $s_n^\pm\ge1$ for all sufficiently large $n$. Applying
Lemma~\ref{lem:K_concentration} to $(s_n^-)$ with parameters $(\varepsilon,\delta_2)$, we
obtain some $n_2$ such that for all $n\ge n_2$,
\[
\mathbb P\left(
\left|
\frac{K(s_n^-)}{n\log\log n}-c_\rho
\right|
\le \varepsilon
\right)\ge 1-\delta_2.
\]
Likewise, applying the same proposition to $(s_n^+)$ with parameters $(\varepsilon,\delta_3)$,
there exists $n_3$ such that for all $n\ge n_3$,
\[
\mathbb P\left(
\left|
\frac{K(s_n^+)}{n\log\log n}-c_\rho
\right|
\le \varepsilon
\right)\ge 1-\delta_3.
\]

For $n\ge \max\{n_1,n_2,n_3\}$, let
\[
A_n^-:=\left\{
\left|
\frac{K(s_n^-)}{n\log\log n}-c_\rho
\right|
\le \varepsilon
\right\}
\quad\text{and}\quad
A_n^+:=\left\{
\left|
\frac{K(s_n^+)}{n\log\log n}-c_\rho
\right|
\le \varepsilon
\right\}.
\]
Then $\mathbb P(E_n)\ge 1-\delta_1$, $\mathbb P(A_n^-)\ge 1-\delta_2$, and
$\mathbb P(A_n^+)\ge 1-\delta_3$.

Consider the good event $G_n:=E_n\cap A_n^-\cap A_n^+$. On $G_n$, the inequalities defining
$A_n^\pm$ give
\[
(c_\rho-\varepsilon)n\log\log n \le K(s_n^-)\le (c_\rho+\varepsilon)n\log\log n
\]
and
\[
(c_\rho-\varepsilon)n\log\log n \le K(s_n^+)\le (c_\rho+\varepsilon)n\log\log n.
\]
Since $E_n$ also implies $K(s_n^-)\le \tau_{\cov}^{(n)}\le K(s_n^+)$, we deduce that on $G_n$,
\[
(c_\rho-\varepsilon)n\log\log n \le \tau_{\cov}^{(n)}\le (c_\rho+\varepsilon)n\log\log n.
\]
Equivalently,
\[
\left|
\frac{\tau_{\cov}^{(n)}}{n\log\log n}-c_\rho
\right|
\le \varepsilon.
\]
Thus
\[
G_n \subseteq
\left\{
\left|
\frac{\tau_{\cov}^{(n)}}{n\log\log n}-c_\rho
\right|
\le \varepsilon
\right\}.
\]

Finally, by the union bound,
\[
\mathbb P(G_n)
\ge 1-\mathbb P(E_n^c)-\mathbb P((A_n^-)^c)-\mathbb P((A_n^+)^c)
\ge 1-\delta_1-\delta_2-\delta_3
=1-\delta.
\]
Therefore, for all $n\ge \max\{n_1,n_2,n_3\}$,
\[
\mathbb P\left(
\left|
\frac{\tau_{\cov}^{(n)}}{n\log\log n}-c_\rho
\right|
\le \varepsilon
\right)\ge 1-\delta.
\]
This proves the theorem.
\end{proof}

Theorem~\ref{thm:star_cover} shows that even in this simple model, exponential self-avoidance
changes the exploration scale from the classical coupon-collector order $n\log n$ to the much smaller
order $n\log\log n$. The remainder of the paper turns from cover time to the main theme of empirical
integral estimation, where a related self-avoidance mechanism yields explicit control of occupation and
transition discrepancies.

\section{Improving SRRW for MCMC Estimation}\label{sec:self_balancing}

Before discussing how TSAW transition rules can be used in a random walk setting, we first analyze the discrete distribution scenario. This will be the building block of a non-Markovian random walk we use to accelerate the MCMC estimation.

Fix $d\ge 2$, a target distribution $p=(p_1,\dots,p_d)$ satisfying
\[
p_i>0,\qquad \sum_{i=1}^d p_i=1,
\qquad p_{\min}:=\min_{i\in[d]}p_i>0,
\]
and a parameter $\lambda>0$.

Let $(X_t)_{t\ge 1}$ be an $[d]:=\{1,\dots,d\}$-valued process adapted to the filtration
$(\mathcal F_t)_{t\ge 0}$, where $\mathcal F_t:=\sigma(X_1,\dots,X_t)$. Define the empirical counts
\[
L_t(i):=\sum_{s=1}^t \mathbf 1\{X_s=i\},
\qquad L_0(i):=0,
\]
and the excess/deficit process
\[
\Delta_i(t):=L_t(i)-tp_i,
\qquad
\Delta(t):=(\Delta_1(t),\dots,\Delta_d(t)).
\]
Then
\[
\sum_{i=1}^d \Delta_i(t)=0
\qquad\text{for all }t\ge 0.
\]

\paragraph{The local TSAW sampling rule.}
Fix $\lambda > 0$.
For $t\ge 0$, define the normalizing factor
\[
Z_t:=\sum_{k=1}^d e^{-\lambda \Delta_k(t)},
\]
and assume that
\begin{equation}\label{eq:rule}
\Pbb(X_{t+1}=i\mid \mathcal F_t)
=
q_t(i)
:=
\frac{e^{-\lambda \Delta_i(t)}}{Z_t},
\qquad i\in[d].
\end{equation}
The deficit recursion is then
\begin{equation}\label{eq:delta_update}
\Delta_i(t+1)=\Delta_i(t)+\mathbf 1\{X_{t+1}=i\}-p_i,
\qquad i\in[d].
\end{equation}

Define the maximum excess and the near-maximum set
\[
M(t):=\max_{i\in[d]} \Delta_i(t),
\qquad
A(t):=\{i\in[d]:\Delta_i(t)\ge M(t)-1\}.
\]
The unit width of \(A(t)\) is chosen to guarantee that if \(X_{t+1}\notin A(t)\), then after the
update in \eqref{eq:delta_update}, the chosen coordinate still cannot exceed the previous maximizers.
The following theorem is the basic Lyapunov estimate for the maximal deficit.

\medskip
\begin{theorem}[Exponential drift for the maximal excess]\label{thm:exp_drift_M}
Fix $\alpha>0$ and define
\[
M(t):=\max_{i\in[d]}\Delta_i(t),
\qquad
A(t):=\{i\in[d]:\Delta_i(t)\ge M(t)-1\},
\qquad
V(t):=e^{\alpha M(t)}.
\]
Let
\[
c_d:=\frac{d}{d-1},
\qquad
\rho:=\exp\Big(-\frac{\alpha p_{\min}}{2}\Big)\in(0,1),
\]
and set
\begin{equation}\label{eq:R_choice_streamlined}
R
:=
\frac{1}{\lambda c_d}\,
\log\Bigg(\frac{d e^{\lambda}(e^{\alpha}-1)}{e^{\alpha p_{\min}/2}-1}\Bigg),
\qquad
B:=\exp\{\alpha(R+1-p_{\min})\}.
\end{equation}
Note that $R > 0$ and $B \ge 1$.
Then for all $t\ge 0$,
\begin{equation}\label{eq:drift_global_streamlined}
\Ebb[V(t+1)\mid \mathcal F_t]\le \rho\,V(t) + B\,\mathbf 1\{M(t)\le R\}.
\end{equation}
Consequently,
\begin{equation}\label{eq:unif_exp_streamlined}
\sup_{t\ge 0}\ \Ebb\big[e^{\alpha M(t)}\big]
\le
\max\Big\{e^{\alpha M(0)},\ \frac{B}{1-\rho}\Big\}<\infty.
\end{equation}
\end{theorem}

\begin{proof}
Fix $t\ge 0$.
Since $\sum_{i=1}^d \Delta_i(t)=0$ and $\max_i \Delta_i(t)=M(t)$, we have
\begin{equation}\label{eq:min_leq_streamlined}
\min_{k\in[d]}\Delta_k(t)\le -\frac{M(t)}{d-1}.
\end{equation}
Indeed, if $\min_k\Delta_k(t)>-M(t)/(d-1)$, then
\[
\sum_{i=1}^d \Delta_i(t)>M(t)+(d-1)\Big(-\frac{M(t)}{d-1}\Big)=0,
\]
contradicting $\sum_i\Delta_i(t)=0$. Therefore
\begin{equation}\label{eq:Z_lower_streamlined}
Z_t=\sum_{k=1}^d e^{-\lambda \Delta_k(t)}
\ge e^{-\lambda\min_k\Delta_k(t)}
\ge \exp\Big(\frac{\lambda}{d-1}M(t)\Big).
\end{equation}

If $i\in A(t)$, then $\Delta_i(t)\ge M(t)-1$, so
$e^{-\lambda \Delta_i(t)}\le e^{-\lambda(M(t)-1)}=e^{\lambda}e^{-\lambda M(t)}$. Summing over $A(t)$ gives
\[
\sum_{i\in A(t)} e^{-\lambda \Delta_i(t)}
\le |A(t)|\,e^{\lambda}e^{-\lambda M(t)}
\le d\,e^{\lambda}e^{-\lambda M(t)}.
\]
Dividing by 
$Z_t$ and using \eqref{eq:rule} and \eqref{eq:Z_lower_streamlined},
we obtain
\begin{equation}\label{eq:pA_streamlined}
\Pbb(X_{t+1}\in A(t)\mid \mathcal F_t)
\le
d\,e^{\lambda}\exp\Big(-\lambda\frac{d}{d-1}M(t)\Big)
=
d\,e^{\lambda}e^{-\lambda c_d M(t)}.
\end{equation}

We next compare $M(t+1)$ with $M(t)$. If $X_{t+1}=j\notin A(t)$, then
$\Delta_j(t)\le M(t)-1$, so by \eqref{eq:delta_update},
\[
\Delta_j(t+1)=\Delta_j(t)+1-p_j\le (M(t)-1)+1-p_{\min}=M(t)-p_{\min}.
\]
For $k\neq j$, we have $\Delta_k(t+1)=\Delta_k(t)-p_k\le M(t)-p_{\min}$. Hence
\begin{equation}\label{eq:M_down_streamlined}
X_{t+1}\notin A(t)\quad\Rightarrow\quad M(t+1)\le M(t)-p_{\min}.
\end{equation}
Also, for every $i\in[d]$,
\[
\Delta_i(t+1)=\Delta_i(t)+\mathbf 1\{X_{t+1}=i\}-p_i\le M(t)+1-p_{\min},
\]
so we always have
\begin{equation}\label{eq:M_up_streamlined}
M(t+1)\le M(t)+(1-p_{\min}).
\end{equation}

Using \eqref{eq:M_down_streamlined} and \eqref{eq:M_up_streamlined}, we get
\[
e^{\alpha M(t+1)}
\le
e^{\alpha(M(t)-p_{\min})}\mathbf 1\{X_{t+1}\notin A(t)\}
+
e^{\alpha(M(t)+1-p_{\min})}\mathbf 1\{X_{t+1}\in A(t)\}.
\]
Taking conditional expectation and factoring out $e^{\alpha(M(t)-p_{\min})}$ gives
\begin{equation}\label{eq:drift_mid_streamlined}
\Ebb[e^{\alpha M(t+1)}\mid \mathcal F_t]
\le
e^{\alpha(M(t)-p_{\min})}
\Big(1+(e^{\alpha}-1)\,\Pbb(X_{t+1}\in A(t)\mid \mathcal F_t)\Big).
\end{equation}
Applying \eqref{eq:pA_streamlined} yields
\begin{equation}    \label{eq:vtplus1bd}
\Ebb[V(t+1)\mid \mathcal F_t]
\le
e^{\alpha M(t)} e^{-\alpha p_{\min}}
\Big(1+d e^{\lambda}(e^{\alpha}-1)e^{-\lambda c_d M(t)}\Big).
\end{equation}

If $M(t)\ge R$, with $R$ given by \eqref{eq:R_choice_streamlined}, then
\[
d e^{\lambda}(e^{\alpha}-1)e^{-\lambda c_d M(t)}
\le
d e^{\lambda}(e^{\alpha}-1)e^{-\lambda c_d R}
=
e^{\alpha p_{\min}/2}-1,
\]
so the 
term in the
bracket 
on the RHS of \eqref{eq:vtplus1bd}
is at most $e^{\alpha p_{\min}/2}$. Therefore
\[
\Ebb[V(t+1)\mid \mathcal F_t]
\le
e^{\alpha M(t)}e^{-\alpha p_{\min}}e^{\alpha p_{\min}/2}
=
\rho\,V(t).
\]

If instead $M(t)\le R$, then \eqref{eq:M_up_streamlined} gives
$M(t+1)\le R+(1-p_{\min})$, hence
\[
V(t+1)\le e^{\alpha(R+1-p_{\min})}=B.
\]
Thus on $\{M(t)\le R\}$ we have
\[
\Ebb[V(t+1)\mid \mathcal F_t]\le B.
\]
Combining the two regimes proves \eqref{eq:drift_global_streamlined}.

Finally, set $a_t:=\Ebb[V(t)]$. Taking expectations in \eqref{eq:drift_global_streamlined} gives
$a_{t+1}\le \rho a_t+B$ for all $t\ge0$. Let
\blue
\footnote{Since $M(0) = 0$ and $B \ge 1$, we have $C = \frac{B}{1-\rho}$, 
but we prefer to write $C$ this way for compatibility with the version of this
calculation that appears when we discuss general graphs rather than just the
star-graph.
}
\black
\[
C:=\max\Big\{e^{\alpha M(0)},\frac{B}{1-\rho}\Big\}.
\]
Then $a_0\le C$, and if $a_t\le C$ then
\[
a_{t+1}\le \rho C+B\le C,
\]
since $C\ge B/(1-\rho)$. By induction, $a_t\le C$ for all $t\ge0$, which proves
\eqref{eq:unif_exp_streamlined}.
\end{proof}

A direct consequence of \eqref{eq:unif_exp_streamlined} is the following uniform-in-$t$ exponential tail bound for each coordinate. Writing
\(
C:=\max\{e^{\alpha M(0)},\frac{B}{1-\rho}\},
\)
Markov's inequality gives
\[
\Pbb\bigl(\Delta_i(t)>x\bigr)\le \Pbb\bigl(M(t)>x\bigr)\le C_\alpha e^{-\alpha x},
\qquad x>0.
\]
On the other hand, since $\sum_{j=1}^d \Delta_j(t)=0$, we have
\(
-\Delta_i(t)\le \sum_{j\neq i}\Delta_j(t)\le (d-1)M(t),
\)
so
\[
\Pbb\bigl(\Delta_i(t)<-x\bigr)\le \Pbb\Bigl(M(t)>\frac{x}{d-1}\Bigr)\le C_\alpha e^{-\alpha x/(d-1)}.
\]
Therefore, for every $i\in[d]$, every $t\ge0$, and every $x>0$,
\[
\Pbb\bigl(|L_t(i)-tp_i|>x\bigr)
=
\Pbb\bigl(|\Delta_i(t)|>x\bigr)
\le
C_\alpha\Bigl(e^{-\alpha x}+e^{-\alpha x/(d-1)}\Bigr)
\le
2C_\alpha e^{-\alpha x/(d-1)}.
\]

And the empirical counts remain within $O(1)$ of their targets $tp_i$ at any finite level of trajectory length $t\geq 0$.

\subsection{Lifting the local TSAW to a random walk}\label{subsec:local_global_model}

Let $V$ be a nonempty finite set, and let $(X_t)_{t\ge 0}$ be a $V$-valued process adapted to the
natural filtration
\[
\mathcal F_t:=\sigma(X_0,X_1,\dots,X_t),\qquad t\ge 0.
\]

Fix a row-stochastic matrix $P=(P_{ij})_{i,j\in V}$:
\[
P_{ij}\ge 0,\qquad \sum_{j\in V}P_{ij}=1\quad\text{for all }i\in V.
\]
Throughout this subsection, irreducibility is assumed, but reversibility is not.

For $t\ge 0$ define the departure counts and directed edge counts
\[
L_t(i):=\sum_{s=0}^{t-1}\mathbf 1\{X_s=i\},
\qquad
N_t(i,j):=\sum_{s=0}^{t-1}\mathbf 1\{X_s=i,\ X_{s+1}=j\},
\qquad i,j\in V.
\]
Please note that $L_t$ is slightly unconventional and only records the number of visits up to time $t-1$ instead of time $t$.
Write $L_t$ for the column vector $(L_t(i))_{i\in V}$ and $N_t$ for the matrix $(N_t(i,j))_{i,j\in V}$.

Row sums satisfy
\begin{equation}\label{eq:row_sum_identity}
N_t\mathbf 1=L_t,
\end{equation}
where $\mathbf 1$ is the all-ones column vector. Column sums satisfy, for each $j\in V$,
\begin{equation}\label{eq:col_flow_identity}
\sum_{i\in V}N_t(i,j)=L_t(j)-\mathbf 1\{X_0=j\}+\mathbf 1\{X_t=j\},
\end{equation}
equivalently,
\begin{equation}\label{eq:colsum_N}
N_t^\top\mathbf 1=L_t-\mathbf e_{X_0}+\mathbf e_{X_t},
\end{equation}
where $\mathbf e_x$ is the standard basis vector at $x$.

Define for $t\ge 0$ and $i,j\in V$ the edge discrepancy
\begin{equation}\label{eq:def_epsilon}
\epsilon_t(i,j):=N_t(i,j)-L_t(i)\,P_{ij}.
\end{equation}
In matrix form,
\begin{equation}\label{eq:N_decomp}
N_t=\diag(L_t)\,P+\epsilon_t.
\end{equation}
Each row of $\epsilon_t$ sums to zero:
\begin{equation}\label{eq:eps_rowsum0}
\epsilon_t\mathbf 1=\mathbf 0,
\end{equation}
since $\sum_j N_t(i,j)=L_t(i)$ and $\sum_j P_{ij}=1$.

\paragraph{The true self-avoiding walk variant.}
Fix $\lambda>0$. The (time-inhomogeneous) transition rule is
\begin{equation}\label{eq:Qt_def_generalP}
\Pbb(X_{t+1}=j\mid \mathcal F_t,\ X_t=i)
=
Q_t(i,j)
:=
\frac{P_{ij}\exp\{-\lambda\,\epsilon_t(i,j)\}}
{\sum_{k\in V}P_{ik}\exp\{-\lambda\,\epsilon_t(i,k)\}}.
\end{equation}
When $P_{ij}=0$ the numerator is $0$, so \eqref{eq:Qt_def_generalP} automatically restricts to the
support $\{j:P_{ij}>0\}$.

For $i\in V$ define the out-neighborhood and out-degree
\[
\mathcal N(i):=\{j\in V:\ P_{ij}>0\},\qquad d_i:=|\mathcal N(i)|.
\]
If $d_i=1$ then $\epsilon_t(i,\cdot)\equiv 0$ deterministically, so the nontrivial case is $d_i\ge2$.

\medskip

For $m\ge 0$ define the $m$-th departure time from row $i$ by
\[
\tau_m^{(i)}:=\inf\{t\ge 0:\ X_t=i,\ L_t(i)=m\},
\]
with the convention $\inf\emptyset=\infty$.

The first step is to isolate a single row \(i\) and observe the process only at successive departure times from that row. 
The resulting process is closely related to the finite-alphabet local TSAW studied earlier, allowing us to exploit the results proved there.

\medskip

\begin{lemma}[Row-wise reduction at actual departure times]\label{lem:row_urn}
Assume $d_i\ge 2$ and $p_{\min}(i):=\min_{j\in\mathcal N(i)}P_{ij}>0$.

For $m\ge 0$ define the \(m\)-th departure time from row \(i\) by
\[
\tau_m^{(i)}:=\inf\{t\ge 0:\ X_t=i,\ L_t(i)=m\},
\]
with the convention \(\inf\emptyset=\infty\). Thus, on the event \(\{\tau_m^{(i)}<\infty\}\), the chain
is at state \(i\) at time \(\tau_m^{(i)}\), and exactly \(m\) departures from \(i\) have occurred
before time \(\tau_m^{(i)}\).

On \(\{\tau_m^{(i)}<\infty\}\) define the destination of the next departure from \(i\) by
\[
Y_{m+1}^{(i)}:=X_{\tau_m^{(i)}+1}.
\]
For \(m\ge 0\) and \(j\in V\), define the embedded row-counts
\[
C_m^{(i)}(j):=\sum_{r=0}^{m-1}\mathbf 1\{\tau_r^{(i)}<\infty,\ Y_{r+1}^{(i)}=j\},
\qquad
\epsilon_m^{(i)}(j):=C_m^{(i)}(j)-mP_{ij},
\]
with the empty sum convention \(C_0^{(i)}(j)=0\).

For \(j\in\mathcal N(i)\), define further, for every \(m\ge0\),
\[
\bar c_i:=\frac{1}{d_i}\sum_{k\in\mathcal N(i)}\frac{1}{\lambda}\log P_{ik},
\qquad
\Delta_m^{(i)}(j):=\epsilon_m^{(i)}(j)-\frac{1}{\lambda}\log P_{ij}+\bar c_i.
\]

Then the following hold.

\begin{enumerate}
\item[\textnormal{(i)}] For every \(t\ge 0\) and \(j\in V\),
\begin{equation}\label{eq:time_change_counts}
N_t(i,j)=C_{L_t(i)}^{(i)}(j),
\qquad
\epsilon_t(i,j)=\epsilon_{L_t(i)}^{(i)}(j).
\end{equation}

\item[\textnormal{(ii)}] For every \(m\ge 0\), on the event \(\{\tau_m^{(i)}<\infty\}\),
\begin{equation}\label{eq:row_urn_rule}
\Pbb\left(Y_{m+1}^{(i)}=j \,\middle|\, \mathcal F_{\tau_m^{(i)}}\right)
=
\frac{P_{ij}\exp\{-\lambda\,\epsilon_m^{(i)}(j)\}}
{\sum_{k\in V}P_{ik}\exp\{-\lambda\,\epsilon_m^{(i)}(k)\}},
\qquad j\in\mathcal N(i).
\end{equation}

\item[\textnormal{(iii)}] For every \(m\ge0\) and every \(j\in\mathcal N(i)\),
\begin{equation}\label{eq:eps_delta_affine}
\epsilon_m^{(i)}(j)=\Delta_m^{(i)}(j)+\frac{1}{\lambda}\log P_{ij}-\bar c_i.
\end{equation}
Moreover, on \(\{\tau_m^{(i)}<\infty\}\),
\[
\sum_{j\in\mathcal N(i)}\Delta_m^{(i)}(j)=0,
\]
and on \(\{\tau_m^{(i)}<\infty\}\) the transition rule becomes
\begin{equation}\label{eq:row_softmax_reduced}
\Pbb\left(Y_{m+1}^{(i)}=j \,\middle|\, \mathcal F_{\tau_m^{(i)}}\right)
=
\frac{e^{-\lambda \Delta_m^{(i)}(j)}}{\sum_{k\in\mathcal N(i)}e^{-\lambda \Delta_m^{(i)}(k)}},
\qquad j\in\mathcal N(i),
\end{equation}
and
\[
\Delta_{m+1}^{(i)}(j)=\Delta_m^{(i)}(j)+\mathbf 1\{Y_{m+1}^{(i)}=j\}-P_{ij}
\qquad\text{on }\{\tau_m^{(i)}<\infty\}.
\]
\end{enumerate}
\end{lemma}

\begin{proof}
For \textnormal{(i)}, note that \(L_t(i)\) is exactly the number of departures from \(i\) before time
\(t\), and the \(r\)-th such departure, when it exists, occurs at time \(\tau_{r-1}^{(i)}\) and goes
to \(Y_r^{(i)}\). Therefore the number of departures from \(i\) to \(j\) before time \(t\) is precisely
the number of indices \(r\in\{1,\dots,L_t(i)\}\) for which \(Y_r^{(i)}=j\), which gives the first
identity in \eqref{eq:time_change_counts}. The second then follows immediately from the definitions of
\(\epsilon_t(i,j)\) and \(\epsilon_m^{(i)}(j)\).

For \textnormal{(ii)}, fix \(m\ge0\) and \(j\in\mathcal N(i)\). First, \(\tau_m^{(i)}\) is a stopping
time with respect to \((\mathcal F_t)\), because for each \(t\ge0\),
\[
\{\tau_m^{(i)}\le t\}
=
\bigcup_{r=0}^t \{X_r=i,\ L_r(i)=m\}\in\mathcal F_t.
\]
Now let \(A\in\mathcal F_{\tau_m^{(i)}}\). Then \(A\cap\{\tau_m^{(i)}=r\}\in\mathcal F_r\) for every
\(r\ge0\). Therefore
\begin{align*}
&\Ebb\!\left[
\mathbf 1_A\,\mathbf 1\{\tau_m^{(i)}<\infty,\ Y_{m+1}^{(i)}=j\}
\right] \\
&\qquad=
\sum_{r=0}^\infty
\Ebb\!\left[
\mathbf 1_{A\cap\{\tau_m^{(i)}=r\}}\mathbf 1\{X_{r+1}=j\}
\right] \\
&\qquad=
\sum_{r=0}^\infty
\Ebb\!\left[
\mathbf 1_{A\cap\{\tau_m^{(i)}=r\}}
\Ebb\!\left[\mathbf 1\{X_{r+1}=j\}\mid \mathcal F_r\right]
\right].
\end{align*}
By \eqref{eq:Qt_def_generalP},
\[
\Ebb\!\left[\mathbf 1\{X_{r+1}=j\}\mid \mathcal F_r\right]
=
Q_r(X_r,j).
\]
On the event \(\{\tau_m^{(i)}=r\}\), we have \(X_r=i\), and by \textnormal{(i)},
\[
\epsilon_r(i,\ell)=\epsilon_m^{(i)}(\ell)\qquad\text{for every }\ell\in V.
\]
Hence on \(\{\tau_m^{(i)}=r\}\),
\[
Q_r(X_r,j)
=
\frac{P_{ij}\exp\{-\lambda\,\epsilon_m^{(i)}(j)\}}
{\sum_{k\in V}P_{ik}\exp\{-\lambda\,\epsilon_m^{(i)}(k)\}}.
\]
Substituting this back yields
\[
\Ebb\!\left[
\mathbf 1_A\,\mathbf 1\{\tau_m^{(i)}<\infty,\ Y_{m+1}^{(i)}=j\}
\right]
=
\Ebb\!\left[
\mathbf 1_A\,\mathbf 1\{\tau_m^{(i)}<\infty\}
\frac{P_{ij}\exp\{-\lambda\,\epsilon_m^{(i)}(j)\}}
{\sum_{k\in V}P_{ik}\exp\{-\lambda\,\epsilon_m^{(i)}(k)\}}
\right],
\]
which proves \eqref{eq:row_urn_rule} on \(\{\tau_m^{(i)}<\infty\}\).

For \textnormal{(iii)}, identity \eqref{eq:eps_delta_affine} is immediate from the definition of
\(\Delta_m^{(i)}\). Now fix \(m\ge0\) and work on the event \(\{\tau_m^{(i)}<\infty\}\). Then the first
\(m\) departures from row \(i\) have all occurred, so for each \(r\in\{0,\dots,m-1\}\), the variable
\(Y_{r+1}^{(i)}\) is defined and takes exactly one value in \(\mathcal N(i)\). Therefore
\[
\sum_{j\in\mathcal N(i)} C_m^{(i)}(j)=m,
\qquad
\sum_{j\in\mathcal N(i)}\epsilon_m^{(i)}(j)=0.
\]
Since \(\sum_{j\in\mathcal N(i)}P_{ij}=1\), it follows that 
on the event \(\{\tau_{m-1}^{(i)}<\infty\}\), and hence on \(\{\tau_m^{(i)}<\infty\}\), we have
\[
\sum_{j\in\mathcal N(i)}\Delta_m^{(i)}(j)
=
\sum_{j\in\mathcal N(i)}
\left(\epsilon_m^{(i)}(j)-\frac{1}{\lambda}\log P_{ij}+\bar c_i\right)
=0.
\]
Also,
\[
e^{-\lambda \Delta_m^{(i)}(j)}
=
e^{-\lambda\bar c_i}\,P_{ij}\,e^{-\lambda\epsilon_m^{(i)}(j)},
\]
so the common factor \(e^{-\lambda\bar c_i}\) cancels in the softmax, giving
\eqref{eq:row_softmax_reduced}.
so multiplying the factor \(e^{-\lambda\bar c_i}\) to both the numerator and the denominator in \eqref{eq:row_urn_rule}, we arrive at \eqref{eq:row_softmax_reduced}.
Finally, on \(\{\tau_m^{(i)}<\infty\}\),
\[
C_{m+1}^{(i)}(j)=C_m^{(i)}(j)+\mathbf 1\{Y_{m+1}^{(i)}=j\},
\]
hence
\[
\epsilon_{m+1}^{(i)}(j)=\epsilon_m^{(i)}(j)+\mathbf 1\{Y_{m+1}^{(i)}=j\}-P_{ij},
\]
and therefore
\[
\Delta_{m+1}^{(i)}(j)=\Delta_m^{(i)}(j)+\mathbf 1\{Y_{m+1}^{(i)}=j\}-P_{ij}.
\]
Please note that on the event \(\{\tau_m^{(i)}<\infty\}\), the variable
\(Y_{m+1}^{(i)}\) is indeed well-defined: indeed, at time \(\tau_m^{(i)}\)
the walk is at \(i\) and exactly \(m\) departures from \(i\) have already
occurred, so the transition from \(\tau_m^{(i)}\) to \(\tau_m^{(i)}+1\)
is precisely the \((m+1)\)-st departure from \(i\).
\qedhere
\end{proof}

\medskip

Once the row-wise reduction has been identified, the finite-alphabet drift theorem, Theorem \ref{thm:exp_drift_M}, applies to the
embedded row process on the event that the next departure from row \(i\) is defined. This yields the
same exponential-drift inequality as in the local model, with indicators keeping track of the
possibility that the row process may terminate.

\medskip
\begin{corollary}[Row-wise exponential moment bound]\label{cor:row_urn_exp_moment}
Under the assumptions of Lemma~\ref{lem:row_urn}, for every \(\alpha>0\) there exist constants
\(\rho_i(\alpha)\in(0,1)\) and \(B_i(\alpha)<\infty\), depending only on
\(\alpha,\lambda\), and the row \(P_{i\cdot}\), such that, writing
\[
M_{\rm urn}^{(i)}(m):=\max_{j\in\mathcal N(i)}\Delta_m^{(i)}(j),
\qquad
V_{\rm urn}^{(i)}(m):=\exp\{\alpha M_{\rm urn}^{(i)}(m)\},
\]
one has, for every \(m\ge0\),
\begin{equation}\label{eq:urn_drift_row}
\Ebb\left[
V_{\rm urn}^{(i)}(m+1)\,\mathbf 1\{\tau_m^{(i)}<\infty\}
\,\middle|\, \mathcal F_{\tau_m^{(i)}}
\right]
\le
\rho_i(\alpha)\,V_{\rm urn}^{(i)}(m)\,\mathbf 1\{\tau_m^{(i)}<\infty\}
+
B_i(\alpha)\,\mathbf 1\{\tau_m^{(i)}<\infty\}.
\end{equation}
Moreover, with the convention \(\tau_{-1}^{(i)}:=0\),
\begin{equation}\label{eq:urn_unif_exp_row}
\sup_{m\ge 0}\ \Ebb\left[
\exp\Big\{\alpha\,\max_{j\in\mathcal N(i)}\Delta_m^{(i)}(j)\Big\}
\,\mathbf 1\{\tau_{m-1}^{(i)}<\infty\}
\right]<\infty.
\end{equation}
In fact one may take
\[
\rho_i(\alpha):=\exp\Big(-\frac{\alpha p_{\min}(i)}{2}\Big),
\qquad
c_i:=\frac{d_i}{d_i-1},
\]
\[
R_i(\alpha):=
\frac{1}{\lambda c_i}
\log\Bigg(\frac{d_i e^{\lambda}(e^{\alpha}-1)}{e^{\alpha p_{\min}(i)/2}-1}\Bigg),
\qquad
B_i(\alpha):=\exp\{\alpha(R_i(\alpha)+1-p_{\min}(i))\}.
\]
\end{corollary}
Note that $\rho_i(\alpha) \in (0,1)$, $R_i(\alpha) > 0$ and $B_i(\alpha) > 1$. Also,
since we are making the assumptions of Lemma~\ref{lem:row_urn}, we are in a scenario where
$d_i \ge 2$, so $c_i$ is well-defined.

\begin{proof}
Fix \(m\ge0\). We emphasize that the row process here is defined using
directed transition counts. On the event \(\{\tau_m^{(i)}<\infty\}\), the
walk is at \(i\), exactly \(m\) departures from \(i\) have already occurred,
and the next transition from \(i\) increases exactly one of the directed
counts \(N_t(i,j)\), \(j\in\mathcal N(i)\). What happens after this departure,
including the path by which the walk may later return to \(i\), plays no role
in this one-step row update.

By Lemma~\ref{lem:row_urn}\textnormal{(iii)}, the vector
\[
\bigl(\Delta_m^{(i)}(j)\bigr)_{j\in\mathcal N(i)}
\]
has zero sum, and conditional on \(\mathcal F_{\tau_m^{(i)}}\) and on
\(\{\tau_m^{(i)}<\infty\}\), the law of the next destination
\(Y_{m+1}^{(i)}=X_{\tau_m^{(i)}+1}\) is
\[
\Pbb\left(Y_{m+1}^{(i)}=j \mid \mathcal F_{\tau_m^{(i)}}\right)
=
\frac{\exp\{-\lambda \Delta_m^{(i)}(j)\}}
{\sum_{k\in\mathcal N(i)}\exp\{-\lambda \Delta_m^{(i)}(k)\}},
\qquad j\in\mathcal N(i).
\]
This is precisely the finite-alphabet softmax rule \eqref{eq:rule}, with
alphabet \(\mathcal N(i)\) and target distribution \(p_j=P_{ij}\). Therefore
the one-step estimate from Theorem~\ref{thm:exp_drift_M} applies to this
embedded row process, yielding \eqref{eq:urn_drift_row}.

Now define
\[
a_m:=\Ebb\left[V_{\rm urn}^{(i)}(m)\,\mathbf 1\{\tau_{m-1}^{(i)}<\infty\}\right],
\qquad m\ge0,
\]
with the convention \(\tau_{-1}^{(i)}:=0\). Since
\[
\Delta_0^{(i)}(j)=-\frac{1}{\lambda}\log P_{ij}+\bar c_i,
\]
we have \(a_0=e^{\alpha M_{\rm urn}^{(i)}(0)}<\infty\). Taking expectations in
\eqref{eq:urn_drift_row} gives
\[
a_{m+1}
=
\Ebb\left[V_{\rm urn}^{(i)}(m+1)\,\mathbf 1\{\tau_m^{(i)}<\infty\}\right]
\le
\rho_i(\alpha)\,
\Ebb\left[V_{\rm urn}^{(i)}(m)\,\mathbf 1\{\tau_m^{(i)}<\infty\}\right]
+
B_i(\alpha)\,\Pbb(\tau_m^{(i)}<\infty).
\]
Since \(\{\tau_m^{(i)}<\infty\}\subseteq\{\tau_{m-1}^{(i)}<\infty\}\), it follows that
\[
a_{m+1}\le \rho_i(\alpha)a_m+B_i(\alpha).
\]
Hence
\[
a_m\le
\max\left\{e^{\alpha M_{\rm urn}^{(i)}(0)},\ \frac{B_i(\alpha)}{1-\rho_i(\alpha)}\right\}
\qquad\text{for all }m\ge0,
\]
which is exactly \eqref{eq:urn_unif_exp_row}.
\end{proof}

\medskip

The next step is to pass from urn time along a single row to a deterministic global time \(t\). The
idea is to condition on the realized number of departures from each row up to time \(t\), invoke the
row-wise exponential moment bound, and then take a union bound over the finitely many rows. 

\medskip
\begin{lemma}[Fixed-time tail bound for the global edge discrepancy]\label{lem:global_time_tail}
Assume that for every \(i\in V\), either \(d_i=1\) or \(d_i\ge 2\) with
\(p_{\min}(i):=\min_{j\in\mathcal N(i)}P_{ij}>0\). 
If \(d_{\max}:=\max_{i\in V}d_i=1\), then
\(M(t)\equiv 0\) for all \(t\). Otherwise define
\[
p_{\min}:=\min_{i\in V}\min_{j\in\mathcal N(i)}P_{ij}\in(0,1).
\]
Then there exist finite constants
\[
C_0:=\log\Big(1+\frac{2}{p_{\min}}\Big),
\qquad
C_1:=
1-p_{\min}
+\frac{1}{\lambda}
\left(
\log d_{\max}+\lambda+\log 2+2\log\frac{1}{p_{\min}}
\right),
\]
such that, for every \(\alpha\ge1\), the quantity
\[
H_\alpha:=\exp\left(\frac{\alpha^2}{\lambda}+C_1\alpha+C_0\right)
\]
results in the bound
\[
\Pbb(M(t)\ge x)\le |V|\,(t+1)\,H_\alpha\,e^{-\alpha x}
\qquad\text{for all }t\ge0\text{ and }x\ge0.
\]
\end{lemma}

\begin{proof}
If \(d_{\max}=1\), then the graph is just a single edge and the walk reduces to a random walk.
We then have \(\epsilon_t(i,\cdot)\equiv0\)
for every \(i\), and hence \(M(t)\equiv0\). Thus suppose \(d_{\max}\ge2\).

Fix \(i\in V\). If \(d_i=1\), define \(M_i(t):=0\) for all \(t\). If \(d_i\ge2\), define
\[
M_i(t):=\max_{j\in\mathcal N(i)}\epsilon_t(i,j).
\]
Then
\[
M(t)=\max_{i\in V}M_i(t).
\]

Fix now a row \(i\) with \(d_i\ge2\). By Lemma~\ref{lem:row_urn}\textnormal{(i)},
\[
M_i(t)=\max_{j\in\mathcal N(i)}\epsilon_t(i,j)
=
\max_{j\in\mathcal N(i)}\epsilon_{L_t(i)}^{(i)}(j).
\]
By Lemma~\ref{lem:row_urn}\textnormal{(iii)}, for every \(m\ge0\) and \(j\in\mathcal N(i)\),
\[
\epsilon_m^{(i)}(j)=\Delta_m^{(i)}(j)+\frac{1}{\lambda}\log P_{ij}-\bar c_i.
\]
Define
\[
K_i:=\max_{j\in\mathcal N(i)}
\left(\frac{1}{\lambda}\log P_{ij}-\bar c_i\right).
\]
Then for every \(m\ge0\),
\[
\max_{j\in\mathcal N(i)}\epsilon_m^{(i)}(j)
\le
\max_{j\in\mathcal N(i)}\Delta_m^{(i)}(j)+K_i.
\]
Therefore, by Corollary~\ref{cor:row_urn_exp_moment},
\begin{equation}\label{eq:eps_exp_row_bound}
\sup_{m\ge0}
\Ebb\left[
\exp\Big\{\alpha\max_{j\in\mathcal N(i)}\epsilon_m^{(i)}(j)\Big\}
\,\mathbf 1\{\tau_{m-1}^{(i)}<\infty\}
\right]
\le
e^{\alpha K_i}
\sup_{m\ge0}
\Ebb\left[
\exp\Big\{\alpha\max_{j\in\mathcal N(i)}\Delta_m^{(i)}(j)\Big\}
\,\mathbf 1\{\tau_{m-1}^{(i)}<\infty\}
\right].
\end{equation}

From Corollary~\ref{cor:row_urn_exp_moment},
\[
\sup_{m\ge0}
\Ebb\left[
\exp\Big\{\alpha\max_{j\in\mathcal N(i)}\Delta_m^{(i)}(j)\Big\}
\,\mathbf 1\{\tau_{m-1}^{(i)}<\infty\}
\right]
\le
\max\left\{e^{\alpha M_{\rm urn}^{(i)}(0)},\ \frac{B_i(\alpha)}{1-\rho_i(\alpha)}\right\},
\]
where
\[
\rho_i(\alpha):=\exp\Big(-\frac{\alpha p_{\min}(i)}{2}\Big),
\qquad
R_i(\alpha):=
\frac{1}{\lambda c_i}
\log\Bigg(\frac{d_i e^{\lambda}(e^{\alpha}-1)}{e^{\alpha p_{\min}(i)/2}-1}\Bigg),
\qquad
B_i(\alpha):=\exp\{\alpha(R_i(\alpha)+1-p_{\min}(i))\}.
\]

Because \(P_{ij}\ge p_{\min}\) on every support edge, every number
\(\lambda^{-1}\log P_{ij}\) lies in the interval
\([-\lambda^{-1}\log(1/p_{\min}),\,0]\). Since \(\bar c_i\) is the average of these numbers,
\[
K_i\le \frac{1}{\lambda}\log\frac1{p_{\min}},
\qquad
M_{\rm urn}^{(i)}(0)
=
\max_{j\in\mathcal N(i)}
\left(-\frac{1}{\lambda}\log P_{ij}+\bar c_i\right)
\le
\frac{1}{\lambda}\log\frac1{p_{\min}}.
\]
Also, since \(d_i\le d_{\max}\), \(p_{\min}(i)\ge p_{\min}\), \(c_i\ge1\), and \(\alpha\ge1\),
\[
R_i(\alpha)\le
\frac{1}{\lambda}
\log\Bigg(\frac{d_{\max}e^{\lambda}(e^{\alpha}-1)}{e^{\alpha p_{\min}/2}-1}\Bigg)
\le
\frac{\alpha}{\lambda}
+
\frac{1}{\lambda}
\left(
\log d_{\max}+\lambda+\log 2+\log\frac1{p_{\min}}
\right),
\]
where we used \(e^\alpha-1\le e^\alpha\) and \(e^u-1\ge u\) for \(u>0\). Moreover,
\[
\frac{1}{1-\rho_i(\alpha)}
=
\frac{1}{1-e^{-\alpha p_{\min}(i)/2}}
\le
1+\frac{2}{\alpha p_{\min}}
\le
1+\frac{2}{p_{\min}},
\]
since \(\alpha\ge1\). Combining these estimates yields
\[
e^{\alpha K_i}
\max\left\{e^{\alpha M_{\rm urn}^{(i)}(0)},\ \frac{B_i(\alpha)}{1-\rho_i(\alpha)}\right\}
\le
\exp\left(\frac{\alpha^2}{\lambda}+C_1\alpha+C_0\right)
=
H_\alpha.
\]
Together with \eqref{eq:eps_exp_row_bound}, this gives
\begin{equation}\label{eq:uniform_row_eps_bound}
\sup_{m\ge0}
\Ebb\left[
\exp\Big\{\alpha\max_{j\in\mathcal N(i)}\epsilon_m^{(i)}(j)\Big\}
\,\mathbf 1\{\tau_{m-1}^{(i)}<\infty\}
\right]
\le H_\alpha
\end{equation}
for every row \(i\) with \(d_i\ge2\).

Now fix \(t\ge0\), \(x\ge0\), and a row \(i\) with \(d_i\ge2\). Since \(L_t(i)\in\{0,1,\dots,t\}\),
\[
\{M_i(t)\ge x\}
=
\bigcup_{m=0}^t
\left(
\{L_t(i)=m\}\cap
\left\{\max_{j\in\mathcal N(i)}\epsilon_m^{(i)}(j)\ge x\right\}
\right),
\]
and these events are disjoint. Therefore
\[
\Pbb(M_i(t)\ge x)
=
\sum_{m=0}^t
\Pbb\Bigl(
L_t(i)=m,\ \max_{j\in\mathcal N(i)}\epsilon_m^{(i)}(j)\ge x
\Bigr).
\]

If \(m=0\), then \(\{L_t(i)=0\}\) implies that no departure from \(i\) has occurred before time \(t\),
hence \(N_t(i,j)=0\) for all \(j\), so \(\epsilon_t(i,j)=0\) for all \(j\), and therefore
\(M_i(t)=0\). Thus the \(m=0\) term is zero for \(x>0\), and is trivially bounded by
\(H_\alpha e^{-\alpha x}\) for \(x=0\).

If \(m\ge1\), then
\[
\{L_t(i)=m\}\subseteq\{\tau_{m-1}^{(i)}<\infty\},
\]
because having \(m\) departures from \(i\) before time \(t\) implies that the \((m-1)\)-st departure
time from \(i\) exists. Hence, by Markov's inequality and \eqref{eq:uniform_row_eps_bound},
\begin{align*}
&\Pbb\Bigl(
L_t(i)=m,\ \max_{j\in\mathcal N(i)}\epsilon_m^{(i)}(j)\ge x
\Bigr) \\
&\qquad\le
e^{-\alpha x}
\Ebb\left[
\exp\Big\{\alpha\max_{j\in\mathcal N(i)}\epsilon_m^{(i)}(j)\Big\}
\,\mathbf 1\{L_t(i)=m\}
\right] \\
&\qquad\le
e^{-\alpha x}
\Ebb\left[
\exp\Big\{\alpha\max_{j\in\mathcal N(i)}\epsilon_m^{(i)}(j)\Big\}
\,\mathbf 1\{\tau_{m-1}^{(i)}<\infty\}
\right] \\
&\qquad\le
H_\alpha e^{-\alpha x}.
\end{align*}
Summing over \(m=0,\dots,t\), we obtain
\[
\Pbb(M_i(t)\ge x)\le (t+1)H_\alpha e^{-\alpha x}
\qquad\text{for every row }i\text{ with }d_i\ge2.
\]

For rows with \(d_i=1\), we have \(M_i(t)\equiv0\), so the same bound is trivial. Therefore, for all
\(i\in V\),
\[
\Pbb(M_i(t)\ge x)\le (t+1)H_\alpha e^{-\alpha x}.
\]
Finally, since \(M(t)=\max_{i\in V}M_i(t)\), a union bound over \(i\in V\) gives
\[
\Pbb(M(t)\ge x)\le |V|\,(t+1)\,H_\alpha\,e^{-\alpha x},
\]
as claimed.
\end{proof}

\medskip

The almost-sure envelope is now obtained by choosing a time-dependent exponential parameter
\(\alpha_t\), applying the fixed-time tail bound with target error probability \(\eta_t\), and then
using Borel--Cantelli.

\medskip
\begin{corollary}
\label{cor:sqrt_log}
Assume the hypotheses of Lemma~\ref{lem:global_time_tail}. If \(d_{\max}:=\max_{i\in V}d_i=1\), then
\(M(t)\equiv 0\) for all \(t\). Otherwise define
\[
p_{\min}:=\min_{i:\,d_i\ge 2}\min_{j\in\mathcal N(i)}P_{ij}\in(0,1],
\qquad
d_{\max}:=\max_{i\in V} d_i.
\]
Fix \(\delta>0\) and set
\[
\eta_t:=(t+2)^{-(1+\delta)},
\qquad
\Lambda_t:=\log\frac{|V|\,(t+1)}{\eta_t}=\log\big(|V|\,(t+1)(t+2)^{1+\delta}\big),
\qquad
\alpha_t:=\max\{1,\sqrt{\lambda \Lambda_t}\}.
\]
Then, almost surely, for all sufficiently large \(t\),
\begin{equation}\label{eq:Mt_as_sqrtlog_selfcontained}
M(t)\le C_1+2\sqrt{\frac{\Lambda_t}{\lambda}}+\frac{C_0}{\alpha_t},
\end{equation}
where \(C_0\) and \(C_1\) are the constants from Lemma~\ref{lem:global_time_tail}.
\end{corollary}

\begin{proof}
If \(d_{\max}=1\), then \(M(t)\equiv 0\), so there is nothing to prove. Assume \(d_{\max}\ge2\).
Apply Lemma~\ref{lem:global_time_tail} with \(\alpha=\alpha_t\ge1\). Since
\(H_{\alpha_t}=\exp(\alpha_t^2/\lambda+C_1\alpha_t+C_0)\), we obtain
\[
\Pbb(M(t)\ge x)\le
|V|\,(t+1)\exp\left(\frac{\alpha_t^2}{\lambda}+C_1\alpha_t+C_0-\alpha_t x\right).
\]
Now set
\[
x_t:=\frac{\alpha_t}{\lambda}+C_1+\frac{C_0}{\alpha_t}+\frac{\Lambda_t}{\alpha_t}.
\]
Because \(\Lambda_t=\log\big(|V|\,(t+1)/\eta_t\big)\), the preceding display becomes
\(
\Pbb(M(t)\ge x_t)\le \eta_t.
\)
Since \(\sum_{t\ge0}\eta_t<\infty\), Borel--Cantelli implies that almost surely the event
\(\{M(t)\ge x_t\}\) occurs only finitely often. Hence almost surely, for all sufficiently large \(t\),
\(
M(t)\le x_t.
\)

Finally, for all sufficiently large \(t\) we have \(\alpha_t=\sqrt{\lambda \Lambda_t}\), so
\[
\frac{\alpha_t}{\lambda}+\frac{\Lambda_t}{\alpha_t}=2\sqrt{\frac{\Lambda_t}{\lambda}}.
\]
Substituting this into the definition of \(x_t\) yields
\eqref{eq:Mt_as_sqrtlog_selfcontained}.
\end{proof}

\subsection{Vertex and edge discrepancy bounds}

The previous subsection gives an almost-sure bound on the edge discrepancy
\[
M(t):=\max_{i\in V}\max_{j\in\mathcal N(i)}\epsilon_t(i,j).
\]
We now convert this bound into control of the vertex discrepancy \(L_t-t\pi\). The key point is that
the column-flow identity for \(N_t\) yields a Poisson equation whose forcing term is determined by
\(\epsilon_t\), and hence by \(M(t)\).

\medskip
\begin{lemma}[Poisson equation and deterministic control of the forcing term]
\label{lem:poisson_eps}\label{lem:bt_vs_M}
For each \(t\ge0\), define
\[
b_t:=\epsilon_t^\top\mathbf 1+\mathbf e_{X_0}-\mathbf e_{X_t}.
\]
Then
\begin{equation}\label{eq:poisson_eps}
(I-P^\top)L_t=b_t,
\end{equation}
and moreover \(\mathbf 1^\top b_t=0\). If \(d_{\max}:=\max_{i\in V}d_i\), then for every
\(i,j\in V\),
\[
|\epsilon_t(i,j)|\le (d_{\max}-1)M(t),
\]
and consequently
\[
\|b_t\|_\infty\le |V|\,(d_{\max}-1)\,M(t)+1.
\]
\end{lemma}

\begin{proof}
By \eqref{eq:colsum_N}, we have \(N_t^\top\mathbf 1=L_t-\mathbf e_{X_0}+\mathbf e_{X_t}\). Using
\eqref{eq:N_decomp}, this becomes
\[
L_t-\mathbf e_{X_0}+\mathbf e_{X_t}
=
N_t^\top\mathbf 1
=
(\diag(L_t)P+\epsilon_t)^\top\mathbf 1
=
P^\top L_t+\epsilon_t^\top\mathbf 1.
\]
Rearranging gives \eqref{eq:poisson_eps}.

Next, \(\mathbf 1^\top b_t=\mathbf 1^\top\epsilon_t^\top\mathbf 1
+\mathbf 1^\top(\mathbf e_{X_0}-\mathbf e_{X_t})=0\), because
\(\epsilon_t\mathbf 1=\mathbf 0\) by \eqref{eq:eps_rowsum0} and
\(\mathbf 1^\top\mathbf e_{X_0}=\mathbf 1^\top\mathbf e_{X_t}=1\).

To bound \(\epsilon_t\), fix \(i\in V\). If \(d_i=1\), then \(\epsilon_t(i,\cdot)\equiv0\), so there is
nothing to prove. Assume \(d_i\ge2\). Since the \(i\)-th row of \(\epsilon_t\) is supported on
\(\mathcal N(i)\) and sums to zero, if
\[
M_i(t):=\max_{j\in\mathcal N(i)}\epsilon_t(i,j),
\]
then
\[
\min_{j\in\mathcal N(i)}\epsilon_t(i,j)\ge -(d_i-1)M_i(t).
\]
Indeed, if some entry were smaller than \(-(d_i-1)M_i(t)\), then the sum of the row would be strictly
negative. Hence, for every \(j\in\mathcal N(i)\),
\[
|\epsilon_t(i,j)|\le (d_i-1)M_i(t)\le (d_{\max}-1)M(t).
\]
If \(j\notin\mathcal N(i)\), then \(P_{ij}=0\) and \(N_t(i,j)=0\), so \(\epsilon_t(i,j)=0\). Thus
\[
|\epsilon_t(i,j)|\le (d_{\max}-1)M(t)
\qquad\text{for all }i,j\in V.
\]

Finally, for each \(j\in V\),
\[
|b_t(j)|
=
\Big|\sum_{i\in V}\epsilon_t(i,j)+\mathbf 1\{X_0=j\}-\mathbf 1\{X_t=j\}\Big|
\le
\sum_{i\in V}|\epsilon_t(i,j)|+1
\le
|V|\,(d_{\max}-1)\,M(t)+1.
\]
Taking the maximum over \(j\) gives the stated bound on \(\|b_t\|_\infty\).
\end{proof}

\medskip

The linear system \((I-P^\top)x=b\) is singular because \(\pi\) lies in the kernel of \(I-P^\top\).
The correct inverse object is therefore the \emph{group inverse}, which acts on the codimension-one
subspace \(\{b:\mathbf 1^\top b=0\}\).

\medskip
Let \(\pi\) denote the unique stationary distribution of the irreducible kernel \(P\), and define
\[
\mathcal L:=I-P^\top,
\qquad
\Pi:=\pi\,\mathbf 1^\top.
\]

\begin{lemma}[Group inverse and normalized solution of the Poisson equation]
\label{lem:vertex_rep}
The matrix \(\mathcal L+\Pi\) is invertible. Define
\begin{equation}\label{eq:group_inverse_def}
\mathcal L^\# := (\mathcal L+\Pi)^{-1}-\Pi.
\end{equation}
This matrix is the group inverse of \(\mathcal L\), in the sense that
\[
\mathcal L\mathcal L^\#=\mathcal L^\#\mathcal L=I-\Pi,
\qquad
\mathcal L^\#\Pi=\Pi\mathcal L^\#=0.
\]
Moreover, for every \(b\in\mathbb R^V\) satisfying \(\mathbf 1^\top b=0\), the equation
\[
\mathcal Lx=b
\]
has a unique solution satisfying \(\mathbf 1^\top x=0\), and this solution is given by
\[
x=\mathcal L^\# b.
\]
\end{lemma}

\begin{proof}
We first show that \(\mathcal L+\Pi\) is invertible. Suppose \((\mathcal L+\Pi)x=0\). Left-multiplying
by \(\mathbf 1^\top\) gives
\[
0=\mathbf 1^\top(\mathcal L+\Pi)x
=\underbrace{\mathbf 1^\top\mathcal L}_{=\,0}x+\mathbf 1^\top\Pi x
=\mathbf 1^\top x,
\]
because \(\mathbf 1^\top\Pi x=\mathbf 1^\top(\pi\,\mathbf 1^\top x)=\mathbf 1^\top x\). Hence
\(\mathbf 1^\top x=0\), so \(\Pi x=0\), and therefore \(\mathcal Lx=0\). Since \(P\) is irreducible on
the finite set \(V\), \(\ker(\mathcal L)=\mathrm{span}\{\pi\}\). Thus \(x=c\pi\) for some scalar \(c\),
and then \(0=\mathbf 1^\top x=c\,\mathbf 1^\top\pi=c\), so \(x=0\). Hence \(\mathcal L+\Pi\) is
invertible.

Now define \(\mathcal L^\#\) by \eqref{eq:group_inverse_def}. Because \(\mathcal L\Pi=0\) and
\(\Pi\mathcal L=0\), we have
\[
(\mathcal L+\Pi)\Pi=\Pi
\qquad\text{and}\qquad
\Pi(\mathcal L+\Pi)=\Pi.
\]
Multiplying by \((\mathcal L+\Pi)^{-1}\) gives
\[
(\mathcal L+\Pi)^{-1}\Pi=\Pi
\qquad\text{and}\qquad
\Pi(\mathcal L+\Pi)^{-1}=\Pi.
\]
Therefore
\[
\mathcal L^\#\Pi=((\mathcal L+\Pi)^{-1}-\Pi)\Pi=0,
\qquad
\Pi\mathcal L^\#=\Pi((\mathcal L+\Pi)^{-1}-\Pi)=0.
\]
Also,
\[
(\mathcal L+\Pi)\mathcal L^\#
=
(\mathcal L+\Pi)\big((\mathcal L+\Pi)^{-1}-\Pi\big)
=
I-\Pi,
\]
and since \(\Pi\mathcal L^\#=0\), this implies \(\mathcal L\mathcal L^\#=I-\Pi\). Similarly,
\[
\mathcal L^\#(\mathcal L+\Pi)
=
\big((\mathcal L+\Pi)^{-1}-\Pi\big)(\mathcal L+\Pi)
=
I-\Pi,
\]
and because \(\mathcal L^\#\Pi=0\), it follows that \(\mathcal L^\#\mathcal L=I-\Pi\).

Finally, let \(b\in\mathbb R^V\) satisfy \(\mathbf 1^\top b=0\). Then \(\Pi b=\pi(\mathbf 1^\top b)=0\).
Set \(x:=\mathcal L^\# b\). Using \(\mathcal L\mathcal L^\#=I-\Pi\), we obtain
\[
\mathcal Lx=\mathcal L\mathcal L^\# b=(I-\Pi)b=b.
\]
Using \(\Pi\mathcal L^\#=0\), we also have
\[
\Pi x=\Pi\mathcal L^\# b=0.
\]
Since \(\Pi x=\pi(\mathbf 1^\top x)\) and \(\pi\) has strictly positive entries, this implies
\(\mathbf 1^\top x=0\).

It remains to prove uniqueness. Suppose \(x\) and \(y\) both satisfy \(\mathcal Lx=\mathcal Ly=b\) and
\(\mathbf 1^\top x=\mathbf 1^\top y=0\). Then \(z:=x-y\) satisfies \(\mathcal Lz=0\), so
\(z\in\mathrm{span}\{\pi\}\). Thus \(z=c\pi\) for some scalar \(c\), and
\(0=\mathbf 1^\top z=c\,\mathbf 1^\top\pi=c\). Hence \(z=0\), so \(x=y\).
\end{proof}

\medskip

We now apply the previous two lemmas with \(b=b_t\) and \(x=D_t:=L_t-t\pi\). The almost-sure
\(\sqrt{\log t}\) envelope for \(M(t)\) then
immediately yields a corresponding bound on the vertex discrepancy.

\medskip
\begin{corollary}[Vertex and edge discrepancy bounds; estimator error]
\label{cor:LLN_vertex_edge}
Assume:
\begin{enumerate}
\item \(P\) is row-stochastic on finite \(V\) and the adaptive rule \eqref{eq:Qt_def_generalP} holds;
\item for each \(i\in V\), either \(d_i=1\) or \(d_i\ge 2\) with
\(p_{\min}(i):=\min_{j\in\mathcal N(i)}P_{ij}>0\);
\item \(P\) is irreducible.
\end{enumerate}
Let \(\pi\) be the stationary distribution of \(P\), define
\[
D_t:=L_t-t\pi,
\]
and fix \(\delta>0\). Set
\[
\Lambda_t:=\log\big(|V|\,(t+1)(t+2)^{1+\delta}\big),
\qquad
\alpha_t:=\max\{1,\sqrt{\lambda\Lambda_t}\}.
\]
Then, almost surely, for all sufficiently large \(t\),
\begin{enumerate}
\item[\textnormal{(i)}] for every \(i\in V\),
\begin{equation}\label{eq:vertex_discrepancy_main}
|L_t(i)-t\pi_i|
\le
\|\mathcal L^\#\|_{\infty\to\infty}
\bigg(
|V|(d_{\max}-1)\bigg(C_1+2\sqrt{\frac{\Lambda_t}{\lambda}}+\frac{C_0}{\alpha_t}\bigg)+1
\bigg);
\end{equation}

\item[\textnormal{(ii)}] for every \((i,j)\) with \(P_{ij}>0\),
\begin{equation}
\begin{aligned}\label{eq:edge_discrepancy_main}
|N_t(i,j)-t\pi_iP_{ij}|
\le &
P_{ij}\,
\|\mathcal L^\#\|_{\infty\to\infty}
\bigg(
|V|(d_{\max}-1)\bigg(C_1+2\sqrt{\frac{\Lambda_t}{\lambda}}+\frac{C_0}{\alpha_t}\bigg)+1
\bigg) \\
& +
(d_{\max}-1)\bigg(C_1+2\sqrt{\frac{\Lambda_t}{\lambda}}+\frac{C_0}{\alpha_t}\bigg);
\end{aligned}
\end{equation}

\item[\textnormal{(iii)}] for every bounded function \(f:V\to\mathbb R\),
\begin{equation}\label{eq:estimator_error_main}
\left|
\frac1t\sum_{s=0}^{t-1} f(X_s)-\sum_{i\in V}\pi_i f(i)
\right|
\le
\frac{\|f\|_\infty\,|V|}{t}\,
\|\mathcal L^\#\|_{\infty\to\infty}
\bigg(
|V|(d_{\max}-1)\bigg(C_1+2\sqrt{\frac{\Lambda_t}{\lambda}}+\frac{C_0}{\alpha_t}\bigg)+1
\bigg).
\end{equation}
\end{enumerate}
In particular,
\[
L_t(i)-t\pi_i=O(\sqrt{\log t}),
\qquad
N_t(i,j)-t\pi_iP_{ij}=O(\sqrt{\log t}),
\]
almost surely for every \(i\in V\) and every \((i,j)\) with \(P_{ij}>0\), and for every bounded
\(f:V\to\mathbb R\),
\[
\left|
\frac1t\sum_{s=0}^{t-1} f(X_s)-\sum_{i\in V}\pi_i f(i)
\right|
=
O\bigg(\frac{\sqrt{\log t}}{t}\bigg)
\qquad\text{almost surely}.
\]
\end{corollary}

\begin{proof}
By Corollary~\ref{cor:sqrt_log}, almost surely, for all sufficiently large \(t\),
\[
M(t)\le C_1+2\sqrt{\frac{\Lambda_t}{\lambda}}+\frac{C_0}{\alpha_t}.
\]
This is exactly the bound recorded in \eqref{eq:vertex_discrepancy_main} after applying the Poisson
equation estimates below.

Next, Lemma~\ref{lem:poisson_eps} gives
\[
(I-P^\top)L_t=b_t,
\]
while \((I-P^\top)(t\pi)=0\) because \(P^\top\pi=\pi\). Hence, with \(D_t:=L_t-t\pi\),
\[
\mathcal L D_t=(I-P^\top)(L_t-t\pi)=b_t.
\]
Also,
\[
\mathbf 1^\top D_t=\mathbf 1^\top L_t-t\,\mathbf 1^\top\pi=t-t=0.
\]
Therefore Lemma~\ref{lem:vertex_rep} yields
\[
D_t=\mathcal L^\# b_t.
\]
Taking \(\ell^\infty\)-norms and using Lemma~\ref{lem:poisson_eps}, we obtain
\[
\|D_t\|_\infty
\le
\|\mathcal L^\#\|_{\infty\to\infty}\,\|b_t\|_\infty
\le
\|\mathcal L^\#\|_{\infty\to\infty}\bigl(|V|(d_{\max}-1)M(t)+1\bigr).
\]
Substituting the bound on \(M(t)\) from Corollary~\ref{cor:sqrt_log} gives
\eqref{eq:vertex_discrepancy_main}.

For the edge discrepancy, use
\[
N_t(i,j)=L_t(i)P_{ij}+\epsilon_t(i,j),
\]
so
\[
N_t(i,j)-t\pi_iP_{ij}
=
P_{ij}\bigl(L_t(i)-t\pi_i\bigr)+\epsilon_t(i,j).
\]
Hence
\[
|N_t(i,j)-t\pi_iP_{ij}|
\le
P_{ij}\,|L_t(i)-t\pi_i|+|\epsilon_t(i,j)|.
\]
By Lemma~\ref{lem:poisson_eps}, we have
\[
|\epsilon_t(i,j)|\le (d_{\max}-1)M(t),
\]
so \eqref{eq:edge_discrepancy_main} follows from \eqref{eq:vertex_discrepancy_main} and the bound
from Corollary~\ref{cor:sqrt_log}.

Finally,
\[
\frac1t\sum_{s=0}^{t-1} f(X_s)-\sum_{i\in V}\pi_i f(i)
=
\frac1t\sum_{i\in V} f(i)\bigl(L_t(i)-t\pi_i\bigr),
\]
and therefore
\[
\left|
\frac1t\sum_{s=0}^{t-1} f(X_s)-\sum_{i\in V}\pi_i f(i)
\right|
\le
\frac{\|f\|_\infty}{t}\sum_{i\in V}|L_t(i)-t\pi_i|
\le
\frac{\|f\|_\infty\,|V|}{t}\,\|D_t\|_\infty.
\]
Combining this with \eqref{eq:vertex_discrepancy_main} yields
\eqref{eq:estimator_error_main}.

The asymptotic statements are immediate from the fact that
\[
\Lambda_t=\log\big(|V|\,(t+1)(t+2)^{1+\delta}\big)=O(\log t)
\]
and
\[
\alpha_t=\max\{1,\sqrt{\lambda\Lambda_t}\}= O(\sqrt{\log t}).
\]
\end{proof}

It is then a natural consequence of this result that this variant of TSAW on the graph is indeed recurrent.

\begin{corollary}[Recurrence of states and directed edges]
\label{cor:recurrence}
Under the assumptions of Corollary~\ref{cor:LLN_vertex_edge}, almost surely every state
\(i\in V\) is visited infinitely often. Moreover, every directed edge \((i,j)\) with
\(P_{ij}>0\) is traversed infinitely often.
\end{corollary}

\begin{proof}
Let \(\delta>0\) be as in Corollary~\ref{cor:LLN_vertex_edge}. Since \(P\) is irreducible
on the finite state space \(V\), its stationary distribution \(\pi\) is strictly positive:
\(\pi_i>0\) for every \(i\in V\).

By Corollary~\ref{cor:LLN_vertex_edge}, there is an event $\Omega_0$ of probability one on
which, for all sufficiently large \(t\), the bounds \eqref{eq:vertex_discrepancy_main} and
\eqref{eq:edge_discrepancy_main} hold simultaneously for every \(i\in V\) and every support
edge \((i,j)\) with \(P_{ij}>0\). Work on this event $\Omega_0$.

Let
\[
A_t:=
\|\mathcal L^\#\|_{\infty\to\infty}
\bigg(
|V|(d_{\max}-1)\bigg(C_1+2\sqrt{\frac{\Lambda_t}{\lambda}}+\frac{C_0}{\alpha_t}\bigg)+1
\bigg).
\]
Then \eqref{eq:vertex_discrepancy_main} gives \(|L_t(i)-t\pi_i|\le A_t\) for all
sufficiently large \(t\) and every \(i\in V\). Since
\(\Lambda_t=\log\big(|V|(t+1)(t+2)^{1+\delta}\big)=O(\log t)\) and
\(\alpha_t=\max\{1,\sqrt{\lambda\Lambda_t}\}\), we have \(A_t=O(\sqrt{\log t})=o(t)\).
Hence, for every \(i\in V\),
\(L_t(i)\ge t\pi_i-A_t\).
Because \(\pi_i>0\) and \(A_t=o(t)\), the right-hand side tends to \(+\infty\), and therefore
\(L_t(i)\to\infty\) for every \(i\in V\). Thus every state is visited infinitely often.

Next fix a directed edge \((i,j)\) with \(P_{ij}>0\). Define
\[
B_t(i,j):=
P_{ij}A_t
+
(d_{\max}-1)\bigg(C_1+2\sqrt{\frac{\Lambda_t}{\lambda}}+\frac{C_0}{\alpha_t}\bigg).
\]
By \eqref{eq:edge_discrepancy_main}, \(|N_t(i,j)-t\pi_iP_{ij}|\le B_t(i,j)\) for all
sufficiently large \(t\). Again \(B_t(i,j)=O(\sqrt{\log t})=o(t)\). Since \(\pi_i>0\) and
\(P_{ij}>0\), we have \(\pi_iP_{ij}>0\), and therefore
\(N_t(i,j)\ge t\pi_iP_{ij}-B_t(i,j)\to+\infty\).
Thus every directed edge \((i,j)\) with \(P_{ij}>0\) is traversed infinitely often.
\end{proof}

\section{Conclusion}

We studied a class of TSAW-based, history-dependent sampling dynamics on finite state spaces and showed that self-avoidance can be used to improve empirical integral estimation in Markov chain Monte Carlo in a quantitative finite-time sense. Starting from a finite-alphabet softmax-balancing process, we proved uniform exponential-moment bounds for the maximal excess, and then lifted this local mechanism to a non-Markovian random walk driven by a finite irreducible kernel \(P\), in which directed edges are penalized according to empirical overuse relative to the target flow \(L_t(i)P_{ij}\). Through a row-wise urn reduction, an almost-sure \(\sqrt{\log t}\) bound for the global edge discrepancy, and a Poisson-equation representation for the occupation error, we obtained the almost-sure bounds
\(
L_t(i)-t\pi_i = O(\sqrt{\log t}),
N_t(i,j)-t\pi_iP_{ij}=O(\sqrt{\log t}),
\)
for every state \(i\in V\) and every edge \((i,j)\) with \(P_{ij}>0\), and consequently
\(
|\frac1t\sum_{s=0}^{t-1} f(X_s)-\sum_{i\in V}\pi_i f(i)|
=
O(\sqrt{\log t} / t)
\)
almost surely for every bounded observable \(f:V\to\mathbb R\). Beyond these specific bounds, the construction also illuminates a broader design space of non-Markovian random walks: once a self-avoiding local process is designed and analyzed quantitatively, one can lift it to a trajectory-dependent walk on a graph or finite state space and propagate the resulting local control to global occupation and flow discrepancies. In this sense, the paper provides both a rigorous TSAW-based mechanism for improving empirical estimation and a general design paradigm for building non-Markovian random walks from self-avoiding local balancing rules.

\appendix

\section{Proofs for Lemma \ref{lem:one_leaf_birth}.}
\label{apx:one_leaf_birth}

\begin{proof}
Write $S_r:=\sum_{m=0}^{r-1}E_m$, 
with $S_0 := 0$,
where the $E_m$ are independent with
$E_m\sim\mathrm{Exp}(\rho^m)$, so that $N(s)\ge r$ if and only if $S_r\le s$.

To prove \textnormal{(i)}, fix $k\ge0$ and set $r:=i_*(s)+k$. 
Since $s \ge 1$, we have $i_*(s) \ge 1$ and so $r \ge 1$.
For any $\theta>0$, Chernoff's bound gives
\[
\mathbb P(N(s)\ge r)=\mathbb P(S_r\le s)\le e^{\theta s}\prod_{m=0}^{r-1}\frac{\rho^m}{\rho^m+\theta}.
\]
Choose $\theta:=\rho^{\,i_*(s)-1}$. Then $\theta s\in[1,1/\rho)$, hence $e^{\theta s}\le e^{1/\rho}=C_1(\rho)$. Also, for every $m$ we have $\rho^m/(\rho^m+\theta)\le1$, while for $m\ge i_*(s)-1$,
\[
\frac{\rho^m}{\rho^m+\theta}\le \frac{\rho^m}{\theta}=\rho^{\,m-(i_*(s)-1)}.
\]
Therefore
\[
\mathbb P(N(s)\ge r)\le
C_1(\rho)\prod_{m=i_*(s)-1}^{r-1}\rho^{\,m-(i_*(s)-1)}
=
C_1(\rho)\,\rho^{\sum_{\ell=0}^k \ell}
=
C_1(\rho)\,\rho^{k(k+1)/2},
\]
which proves \textnormal{(i)}.

For \textnormal{(ii)}, first suppose $1\le k\le i_*(s)-1$ and set $r:=i_*(s)-k\ge1$. For any
$\theta\in(0,\rho^{r-1})$, Chernoff's bound yields
\[
\mathbb P(N(s)\le r)=\mathbb P(S_r\ge s)\le e^{-\theta s}\prod_{m=0}^{r-1}\frac{\rho^m}{\rho^m-\theta}.
\]
Take $\theta:=\tfrac12\rho^{r-1}$. Then
\[
\prod_{m=0}^{r-1}\frac{\rho^m}{\rho^m-\theta}
=
\prod_{m=0}^{r-1}\frac{1}{1-\frac12\rho^{r-1-m}}
\le
\prod_{m=0}^{\infty}\frac{1}{1-\frac12\rho^m}
=
C_2(\rho).\footnote{This actually is $(\frac12;\rho)_\infty$, where $(a;q)_\infty:=\prod_{m=0}^\infty (1-aq^m)$ denotes the $q$-Pochhammer symbol.}
\]
Moreover, since $\rho^{\,i_*(s)-1}s\in[1,1/\rho)$, we have
\[
\theta s=\frac12\rho^{r-1}s=\frac12\rho^{\,i_*(s)-1-k}s\in\left[\frac12\rho^{-k},\,\frac12\rho^{-k-1}\right),
\]
so \(e^{-\theta s}\le \exp(-\tfrac12\rho^{-k})\le \exp(-\tfrac12\rho^{-k+1})\). Hence
\[
\mathbb P\bigl(N(s)\le i_*(s)-k\bigr)\le
C_2(\rho)\exp\bigl(-\tfrac12\rho^{-k+1}\bigr)
\qquad\text{for }1\le k\le i_*(s)-1.
\]

If \(k=i_*(s)\), then \(N(s)\le i_*(s)-k\) means \(N(s)=0\), so
\(\mathbb P(N(s)\le i_*(s)-k)=\mathbb P(E_0>s)=e^{-s}\). Since \(s\ge \rho^{-i_*(s)+1}\), we get
\(e^{-s}\le e^{-\rho^{-i_*(s)+1}}\le C_2(\rho)\exp(-\tfrac12\rho^{-i_*(s)+1})\), because
\(C_2(\rho)\ge1\). If \(k>i_*(s)\), the event is empty. This proves \textnormal{(ii)}.

For \textnormal{(iii)}, write $i_*:=i_*(s)$. Since $N(s)$ is a nonnegative integer-valued random
variable, we have the pointwise identity
\[
N(s)=\sum_{r=1}^{\infty}\mathbf 1\{N(s)\ge r\}.
\]
Applying Tonelli's theorem yields\footnote{Here we use Tonelli's theorem in the form
\[
\mathbb E\left[\sum_{r=1}^\infty X_r\right]
=
\sum_{r=1}^\infty \mathbb E[X_r]
\]
for any sequence of nonnegative random variables $(X_r)_{r\ge1}$, where both sides are allowed to take the value $+\infty$. Applying this with $X_r=\mathbf 1\{N(s)\ge r\}$ gives
\[
\mathbb E[N(s)]
=
\sum_{r=1}^\infty \mathbb P(N(s)\ge r).
\]
}
\[
\mathbb E[N(s)]
=
\sum_{r=1}^{\infty}\mathbb P(N(s)\ge r).
\]

For the upper bound on the mean, split the sum at $i_*$ and use \textnormal{(i)}:
\[
\mathbb E[N(s)]
=
\sum_{r=1}^{i_*}\mathbb P(N(s)\ge r)
+
\sum_{r=i_*+1}^{\infty}\mathbb P(N(s)\ge r)
\le
i_*+\sum_{k=1}^{\infty}\mathbb P(N(s)\ge i_*+k).
\]
Therefore
\[
\mathbb E[N(s)]
\le
i_*+C_1(\rho)\sum_{k=1}^{\infty}\rho^{k/2}
=
i_*+\frac{C_1(\rho)\sqrt{\rho}}{1-\sqrt{\rho}}
\le
i_*+\frac{C_1(\rho)}{1-\sqrt{\rho}}.
\]

For the lower bound, we again split at $i_*$:
\[
\mathbb E[N(s)]
=
\sum_{r=1}^{i_*}\mathbb P(N(s)\ge r)
+
\sum_{r=i_*+1}^{\infty}\mathbb P(N(s)\ge r).
\]
The second term is nonnegative, so
\[
\mathbb E[N(s)]
\ge
\sum_{r=1}^{i_*}\mathbb P(N(s)\ge r)
=
\sum_{r=1}^{i_*}\bigl(1-\mathbb P(N(s)<r)\bigr)
=
i_*-\sum_{r=1}^{i_*}\mathbb P(N(s)\le r-1).
\]
Reindexing with $k:=i_*-r+1$ gives
\[
\sum_{r=1}^{i_*}\mathbb P(N(s)\le r-1)
=
\sum_{k=1}^{i_*}\mathbb P(N(s)\le i_*-k).
\]
Hence
\[
\mathbb E[N(s)]
\ge
i_*-\sum_{k=1}^{i_*}\mathbb P(N(s)\le i_*-k).
\]
Applying \textnormal{(ii)} and then extending the finite sum to an infinite one, we obtain
\[
\mathbb E[N(s)]
\ge
i_*-\sum_{k=1}^{\infty} C_2(\rho)\exp\bigl(-\tfrac12\rho^{-k+1}\bigr).
\]
Using $e^{-x}\le 1/x$ for $x>0$, we get
\[
\exp\bigl(-\tfrac12\rho^{-k+1}\bigr)\le 2\rho^{k-1},
\]
and therefore
\[
\mathbb E[N(s)]
\ge
i_*-\frac{2C_2(\rho)}{1-\rho}.
\]

Since $\bigl|i_*(s)-\log s/\log(1/\rho)\bigr|\le 1$, combining the two bounds gives
\[
\left|
\mathbb E[N(s)]-\frac{\log s}{\log(1/\rho)}
\right|
\le
1+\frac{C_1(\rho)}{1-\sqrt{\rho}}+\frac{2C_2(\rho)}{1-\rho}
=
A_\rho.
\]

To bound the variance, let $Y_s:=|N(s)-i_*(s)|$. By \textnormal{(i)} and \textnormal{(ii)}, for every
$m\ge1$,
\[
\mathbb P(Y_s\ge m)
\le
C_1(\rho)\rho^{m/2}+C_2(\rho)\exp\bigl(-\tfrac12\rho^{-m+1}\bigr).
\]
Since $Y_s$ is nonnegative and integer-valued, we have the pointwise identity
\[
Y_s^2=\sum_{m=1}^{\infty}(2m-1)\mathbf 1\{Y_s\ge m\}.
\]
Applying Tonelli's theorem again gives
\[
\mathbb E[Y_s^2]
=
\sum_{m=1}^{\infty}(2m-1)\,\mathbb P(Y_s\ge m).
\]
Hence
\[
\mathbb E[Y_s^2]
\le
C_1(\rho)\sum_{m=1}^{\infty}(2m-1)\rho^{m/2}
+
C_2(\rho)\sum_{m=1}^{\infty}(2m-1)\exp\bigl(-\tfrac12\rho^{-m+1}\bigr).
\]
Using again $e^{-x}\le 1/x$, the second sum is bounded by
\[
2C_2(\rho)\sum_{m=1}^{\infty}(2m-1)\rho^{m-1}.
\]

Note that for $|x|<1$, one has
\[
\sum_{m=1}^\infty (2m-1)x^m
=
2\sum_{m=1}^\infty m x^m-\sum_{m=1}^\infty x^m
=
\frac{2x}{(1-x)^2}-\frac{x}{1-x}
=
\frac{x(1+x)}{(1-x)^2}.
\]
Applying this with $x=\sqrt{\rho}$ and $x=\rho$ gives, respectively, 
\[
\sum_{m=1}^{\infty}(2m-1)\rho^{m/2}
=
\frac{\sqrt{\rho}(1+\sqrt{\rho})}{(1-\sqrt{\rho})^2},
\qquad
\sum_{m=1}^{\infty}(2m-1)\rho^m
=
\frac{\rho(1+\rho)}{(1-\rho)^2}.
\]

Therefore
\[
\mathbb E[Y_s^2]
\le
C_1(\rho)\frac{\sqrt{\rho}(1+\sqrt{\rho})}{(1-\sqrt{\rho})^2}
+
2C_2(\rho)\frac{1+\rho}{(1-\rho)^2}.
\]

Finally, since
\[
\operatorname{Var}(N(s))
=
\mathbb E\bigl[(N(s)-\mathbb E[N(s)])^2\bigr]
\le
2\,\mathbb E\bigl[(N(s)-i_*(s))^2\bigr]
+
2\,(\mathbb E[N(s)]-i_*(s))^2,
\]
the first term is bounded by twice the previous display, and the second is at most
\[
2\left(
\frac{C_1(\rho)}{1-\sqrt{\rho}}+\frac{2C_2(\rho)}{1-\rho}
\right)^2.
\]
Combining these bounds yields $\operatorname{Var}(N(s))\le B_\rho$.
\end{proof}

\section{Proof of Lemma~\ref{lem:K_concentration}.}
\label{sec:proof-k-concentration}

\begin{proof}
Since $s_n=\log n+c$, we have $s_n\ge1$ for all sufficiently large $n$. For each fixed $n$, the
random variables $N_1(s_n),\dots,N_n(s_n)$ are i.i.d., so
$\mathbb E[K(s_n)]=n\,\mathbb E[N(s_n)]$ and
$\operatorname{Var}(K(s_n))=n\,\operatorname{Var}(N(s_n))$.

By Lemma~\ref{lem:one_leaf_birth}, for every such $n$,
\[
\left|
\mathbb E[N(s_n)]-\frac{\log s_n}{\log(1/\rho)}
\right|
\le A_\rho
\qquad\text{and}\qquad
\operatorname{Var}(N(s_n))\le B_\rho.
\]
Multiplying the first bound by $n$ and dividing by $n\log\log n$ yields
\[
\left|
\frac{\mathbb E[K(s_n)]}{n\log\log n}
-
c_\rho\,\frac{\log s_n}{\log\log n}
\right|
\le
\frac{A_\rho}{\log\log n}.
\]

We now compare $\log s_n$ with $\log\log n$ explicitly. Since $s_n=\log n+c$, we can write
\[
\log s_n=\log(\log n+c)=\log\log n+\log\left(1+\frac{c}{\log n}\right).
\]
Choose $n$ large enough that $\log n\ge 2|c|$, so that $\left|c/\log n\right|\le 1/2$. Then the
elementary bound $|\log(1+u)|\le 2|u|$ for $|u|\le 1/2$ gives
\[
\left|\log s_n-\log\log n\right|
=
\left|\log\left(1+\frac{c}{\log n}\right)\right|
\le
\frac{2|c|}{\log n}.
\]
Dividing by $\log\log n$, we obtain
\[
\left|
\frac{\log s_n}{\log\log n}-1
\right|
\le
\frac{2|c|}{\log n\,\log\log n}.
\]

Fix $\varepsilon>0$ and $\delta>0$. Since $\log\log n\to\infty$ and
$\log n\,\log\log n\to\infty$, there exists $n_1$ such that for all $n\ge n_1$,
\[
\frac{A_\rho}{\log\log n}\le \frac{\varepsilon}{4}
\qquad\text{and}\qquad
c_\rho\,\frac{2|c|}{\log n\,\log\log n}\le \frac{\varepsilon}{4}.
\]
For such $n$, the two previous displays imply
\[
\left|
\frac{\mathbb E[K(s_n)]}{n\log\log n}-c_\rho
\right|
\le
\left|
\frac{\mathbb E[K(s_n)]}{n\log\log n}
-
c_\rho\,\frac{\log s_n}{\log\log n}
\right|
+
c_\rho\left|
\frac{\log s_n}{\log\log n}-1
\right|
\le
\frac{\varepsilon}{2}.
\]

Similarly,
\[
\operatorname{Var}\left(\frac{K(s_n)}{n\log\log n}\right)
=
\frac{\operatorname{Var}(K(s_n))}{n^2(\log\log n)^2}
=
\frac{\operatorname{Var}(N(s_n))}{n(\log\log n)^2}
\le
\frac{B_\rho}{n(\log\log n)^2}.
\]
Since $n(\log\log n)^2\to\infty$, there exists $n_2$ such that for all $n\ge n_2$,
\[
\frac{4B_\rho}{\varepsilon^2\,n(\log\log n)^2}\le \delta.
\]

Now let $X_n:=K(s_n)/(n\log\log n)$. For $n\ge \max\{n_1,n_2\}$, Chebyshev's inequality gives
\[
\mathbb P\left(
\left|
X_n-\mathbb E[X_n]
\right|>\frac{\varepsilon}{2}
\right)
\le
\frac{4\,\operatorname{Var}(X_n)}{\varepsilon^2}
\le
\frac{4B_\rho}{\varepsilon^2\,n(\log\log n)^2}
\le \delta.
\]
On the event $\{|X_n-\mathbb E[X_n]|\le \varepsilon/2\}$, the bound
$|\mathbb E[X_n]-c_\rho|\le \varepsilon/2$ implies $|X_n-c_\rho|\le \varepsilon$. Therefore
\[
\mathbb P\left(
\left|
\frac{K(s_n)}{n\log\log n}-c_\rho
\right|
\le \varepsilon
\right)\ge 1-\delta
\]
for all $n\ge \max\{n_1,n_2\}$. This proves the lemma.
\end{proof}


\begin{thebibliography}{99}

\bibitem{DoshiHuEun2023}
Doshi, Vishwaraj, Jie Hu, and Do Young Eun. 2023.
``Self-Repellent Random Walks on General Graphs: Achieving Minimal Sampling Variance via Nonlinear Markov Chains.''
In \emph{Proceedings of the 40th International Conference on Machine Learning (ICML 2023)},
PMLR 202.

\bibitem{HuMaEun2025}
Hu, Jie, Yi-Ting Ma, and Do Young Eun. 2025.
``Beyond Self-Repellent Kernels: History-Driven Target Towards Efficient Nonlinear MCMC on General Graphs.''
arXiv preprint arXiv:2505.18300.

\bibitem{AmitParisiPeliti1983}
Amit, Daniel J., Giorgio Parisi, and Luca Peliti. 1983.
``Asymptotic Behavior of the `True' Self-Avoiding Walk.''
\emph{Physical Review B} 27: 1635--1645.

\bibitem{Toth1995}
T\'oth, B\'alint. 1995.
``The `True' Self-Avoiding Walk with Bond Repulsion on $\mathbb{Z}$: Limit Theorems.''
\emph{The Annals of Probability} 23 (4): 1523--1556.

\bibitem{VetoToth2008}
Vet\H{o}, B\'alint, and B\'alint T\'oth. 2008.
``Self-Repelling Random Walk with Directed Edges on $\mathbb{Z}$.'' 
\emph{Electronic Journal of Probability} 13: 1909--1926.

\bibitem{Grassberger2017}
Grassberger, Peter. 2017.
``Self-Trapping Self-Repelling Random Walks.''
\emph{Physical Review Letters} 119: 140601.

\bibitem{DiaconisHolmesNeal2000}
Diaconis, Persi, Susan Holmes, and Radford M. Neal. 2000.
``Analysis of a Nonreversible Markov Chain Sampler.''
\emph{The Annals of Applied Probability} 10: 726--752.

\bibitem{Neal2004}
Neal, Radford M. 2004.
``Improving Asymptotic Variance of MCMC Estimators: Non-Reversible Chains Are Better.''
arXiv preprint math/0407281.

\bibitem{AlonBenjaminiLubetzkySodin2007}
Alon, Noga, Itai Benjamini, Eyal Lubetzky, and Sasha Sodin. 2007.
``Non-Backtracking Random Walks Mix Faster.''
\emph{Communications in Contemporary Mathematics} 9 (4): 585--603.

\bibitem{LeeXuEun2012}
Lee, Chul-Ho, Xin Xu, and Do Young Eun. 2012.
``Beyond Random Walk and Metropolis--Hastings Samplers: Why You Should Not Backtrack for Unbiased Graph Sampling.''
arXiv preprint arXiv:1204.4140.

\bibitem{TuritsynChertkovVucelja2011}
Turitsyn, Konstantin S., Michael Chertkov, and Marija Vucelja. 2011.
``Irreversible Monte Carlo Algorithms for Efficient Sampling.''
\emph{Physica D: Nonlinear Phenomena} 240 (4--5): 410--414.

\bibitem{ChenHwang2013}
Chen, Guan-Yu, and Chii-Ruey Hwang. 2013.
``Accelerating Reversible Markov Chains.''
\emph{Statistics \& Probability Letters} 83 (9): 1956--1962.

\bibitem{MaChenWuFox2016}
Ma, Yi-An, Tianqi Chen, Lei Wu, and Emily B. Fox. 2016.
``A Unifying Framework for Devising Efficient and Irreversible MCMC Samplers.''
arXiv preprint arXiv:1608.05973.

\bibitem{ThinEtAl2020}
Thin, Achille, Nikita Kotelevskii, Alain Durmus, Eric Moulines, Arnaud Doucet, and Pierre Jacob. 2020.
``Nonreversible MCMC from Conditional Invertible Transforms: A Complete Recipe with Convergence Guarantees.''
arXiv preprint arXiv:2012.15550.

\bibitem{DelMoralMiclo2004}
Del Moral, Pierre, and Laurent Miclo. 2004.
``On Convergence of Chains with Occupational Self-Interactions.''
\emph{Proceedings of the Royal Society A} 460 (2041): 325--346.

\bibitem{DelMoralMiclo2006}
Del Moral, Pierre, and Laurent Miclo. 2006.
``Self-Interacting Markov Chains.''
\emph{Stochastic Analysis and Applications} 24: 615--660.

\bibitem{AndrieuJasraDoucetDelMoral2007}
Andrieu, Christophe, Ajay Jasra, Arnaud Doucet, and Pierre Del Moral. 2007.
``Nonlinear Markov Chain Monte Carlo.''
\emph{ESAIM: Proceedings} 19: 79--84.

\bibitem{AndrieuJasraDoucetDelMoral2011}
Andrieu, Christophe, Ajay Jasra, Arnaud Doucet, and Pierre Del Moral. 2011.
``On Nonlinear Markov Chain Monte Carlo.''
\emph{Bernoulli} 17 (3): 987--1014.

\bibitem{DelMoralDoucet2010}
Del Moral, Pierre, and Arnaud Doucet. 2010.
``Interacting Markov Chain Monte Carlo Methods for Solving Nonlinear Measure-Valued Equations.''
\emph{The Annals of Applied Probability} 20 (2): 593--639.

\bibitem{FortMoulinesPriouret2011}
Fort, Gersende, Eric Moulines, and Philippe Priouret. 2011.
``Convergence of Adaptive and Interacting Markov Chain Monte Carlo Algorithms.''
\emph{The Annals of Statistics} 39 (6): 3262--3289.

\bibitem{Fort2015}
Fort, Gersende. 2015.
``Central Limit Theorems for Stochastic Approximation with Controlled Markov Chain Dynamics.''
\emph{ESAIM: Probability and Statistics} 19: 60--80.

\end{thebibliography}
\end{document}